\documentclass[journal]{IEEEtran}

\usepackage{amsmath,amssymb,amsfonts}
\usepackage{amsthm}
\usepackage{algorithm}
\usepackage{algpseudocode}
\usepackage{cite}
\usepackage[T1]{fontenc}
\usepackage[utf8]{inputenc}
\usepackage{graphicx}
\usepackage{subcaption}
\usepackage{textcomp}
\usepackage{xcolor}
\usepackage{booktabs}
\usepackage{bm}
\usepackage{tikz}
\usetikzlibrary{arrows.meta,calc,patterns,positioning}
\newcommand{\vct}[1]{\bm{#1}}
\newcommand{\mtx}[1]{\mathbf{#1}}
\newcommand{\bmtx}[1]{\bm{#1}}
\DeclareMathOperator{\tr}{tr}
\newtheorem{proposition}{Proposition}
\newtheorem{lemma}{Lemma}
\newtheoremstyle{definitionnopunct}
    {\topsep}
    {\topsep}
    {\normalfont}
    {}
    {\bfseries}
    {}
    { }
    {}
\theoremstyle{definitionnopunct}
\newtheorem{definition}{Definition}

\newtheorem{remark}{Remark}

\begin{document}

\bstctlcite{paperBSTcontrol}
\title{Control-Certified Wireless Resource Allocation for Digital-Twin-Enabled UAV Swarms}

\author{
   Qingyun Luo,~\IEEEmembership{Student Member,~IEEE,}
   Jingqing Wang,~\IEEEmembership{Member,~IEEE,}
   and Wenchi Cheng,~\IEEEmembership{Senior Member,~IEEE}%
   \thanks{Q. Luo, J. Wang, and W. Cheng are with the School of Telecommunication Engineering, Xidian University, Xi'an 710071, China. E-mail: luoqy@stu.xidian.edu.cn, \{jqwangxd, wccheng\}@xidian.edu.cn.}%
}



\maketitle
\begin{abstract}
   Wireless resource allocation in digital-twin-enabled unmanned aerial vehicle (UAV) swarms must be both network-feasible and certifiably safe for closed-loop control.
   Existing packet-level or scalar-priority schedulers cannot meaningfully compare heterogeneous multi-hop actions that differ simultaneously in route, retransmission depth, blocklength, bidirectional delay, delivery probability, and TDMA slot cost.
   This paper introduces a certificate-guided resource allocation framework for low-altitude multi-hop UAV swarms.
   A digital twin maps predicted topology, channel, route, and controller-side state into a shared five-dimensional quality-of-service (QoS) certificate comprising uplink/downlink delay bounds, directional delivery guarantees, and a certified upper bound on the interval between successful bidirectional interactions.
   A state-conditioned stochastic drift test then admits only certificates whose augmented Lyapunov drift is nonpositive under the current controller state.
   Admitted actions are reduced to certified supply frontiers by removing dominated route-slot configurations, and the online scheduler maximizes Lyapunov-drift reduction under a shared TDMA slot budget via exact dynamic programming.
   Closed-loop ns-3 simulations demonstrate that the proposed framework outperforms fixed-service, certificate-filtered fixed-priority, dynamic-transmission-count, and value-of-information baselines in both tracking accuracy and high-risk state suppression under identical communication budgets.
\end{abstract}

\begin{IEEEkeywords}
   Digital Twin Networks, Closed-Loop Optimization, Wireless Resource Allocation, UAV Swarms, Communication-Control Co-Design
\end{IEEEkeywords}
\section{Introduction}
\label{sec:introduction}
Wireless resource allocation for closed-loop unmanned aerial vehicle (UAV) swarms differs fundamentally from conventional quality-of-service (QoS) scheduling.
The traffic carried by each link is a bidirectional telemetry--command exchange: uplink telemetry updates the controller-side state estimate, while downlink commands drive the plant before that estimate becomes stale.
Missed or clustered exchanges widen estimation uncertainty and accelerate closed-loop drift; consequently, transmission success, route feasibility, and slot cost matter only through their effect on the control loop~\cite{heemelsNetworkedControlSystems2010,mengCommunicationSensingControl2026,zhouIntegratedSensingCommunication2025}.
In a multi-hop time-division multiple access (TDMA) swarm, the allocator must decide at every control period which bidirectional flows receive finite slots under time-varying topology, channel quality, mission geometry, and control risk~\cite{wangTowardRealizationLowAltitude2025,chaoComputingPowerSky2025}.

Digital twins (DTs) offer a principled interface for this setting because they can forecast topology, channel conditions, route availability, and controller-side state ahead of each allocation decision~\cite{wuDigitalTwinNetworks2021,khanDigitalTwinWireless2022,zhangDigitalNetworkTwins2024}.
Existing wireless DT and task-oriented communication studies exploit twin-side state for edge association, computing offloading, application-driven management, and semantics-aware UAV command delivery~\cite{bellavistaApplicationDriven2021,luAdaptiveEdgeAssociation2021,yangJointCommunicationComputation2024,xuTaskOriented2023}.
Recent digital twin network (DTN) studies further couple twins with vehicular control and integrated sensing--communication--computing--control loops~\cite{leiVoIDrivenJoint2025,fangSensingCommunicationComputingControl2025}.
However, prediction alone does not constitute a scheduling interface.
A DTN scheduler additionally requires a control-grounded characterization that determines whether a predicted route--slot action is admissible for the current loop state and comparable with alternative feasible actions.

Constructing such a characterization is difficult because candidate actions are inherently heterogeneous.
For the same control loop, one action may traverse a shorter route with lower reliability, another may employ deeper retransmission, another may adopt a different coded-slot pattern, and another may occupy a different portion of the TDMA frame~\cite{wangJointRoutingScheduling2024,liuLatencyRateReliability2021}.
The resulting uplink delay, downlink delay, service-success probability, interaction interval, and slot cost lack a common scalar scale.
Communication--control co-design methods can rank transmissions by control cost, state value, age, semantic logic, goal-oriented utility, or policy-level objectives~\cite{maOptimalDynamicTransmission2022,huangOptimalDownlinkUplink2020,ayanAgeInformationValueInformation2019,girgisSemanticLogicalCommunicationControl2023,caoGoalOrientedCommunication2025,pangCommunicationControlCodesignLargeScale2025}, while cross-layer UAV and deterministic-network scheduling methods can enforce route and slot feasibility~\cite{wangJointRoutingScheduling2024,yangMultipolicyDeepReinforcement2024,liMeanFieldGameTheoretic2021}.
Neither family defines when one multi-hop route--slot action safely substitutes for another from the standpoint of closed-loop stability.
As a result, a low-slot action may be network-feasible yet unsafe for the current loop state, whereas a higher-slot action may be the only admissible bidirectional service for that loop.

This paper develops a certificate-preserving DTN resource-allocation method for multi-hop UAV swarms.
The central construct is a \emph{shared QoS certificate}: a five-dimensional supply specification comprising uplink delay, downlink delay, a certified bidirectional interaction interval, and DT-calibrated directional service-success lower bounds.
Intuitively, a certificate determines whether one route--slot action can safely replace another from the standpoint of closed-loop control.
The DT maps predicted topology, channel conditions, route, action, and controller-side state into this certificate under a twin confidence set, and stochastic drift admission tests whether the certificate satisfies fixed closed-loop preconditions.
Because the certificate embeds the interaction interval and directional service-success lower bounds into a partially ordered supply space, a stronger certificate provides no weaker service support and no worse drift bound.
The online allocator prunes certificate--cost dominated actions, retains only the certified supply frontier, and solves the resulting slot-constrained multi-choice knapsack problem exactly via dynamic programming.
This reduction preserves the closed-loop certificate while converting heterogeneous route--slot choices into a tractable allocation problem under the shared TDMA budget.

The method is evaluated in closed-loop ns-3 multi-hop UAV swarm simulations against Fixed-Service, Cert-Fixed, DynTx-HLC, and VoI-Sched baselines.
Certificate-level frontier pruning preserves exact optimality over the reduced action space while improving tracking accuracy and suppressing high-risk Lyapunov-cost events.
These results indicate that the certificate converts DT prediction into schedulable wireless actions with control-level guarantees, rather than merely optimizing transmission latency or scalar priority.

The remainder of this paper is organized as follows.
Section~\ref{sec:system-model} defines the DTN-assisted multi-hop TDMA slot-frame model and the closed-loop preconditions.
Section~\ref{sec:certificate-admission} develops the shared QoS certificate and the closed-loop admission test.
Section~\ref{sec:frontier-allocation} presents the certified supply frontier and the exact slot-constrained allocation algorithm.
Section~\ref{sec:evaluation} evaluates the method in closed-loop ns-3 swarm simulations, and Section~\ref{sec:conclusion} concludes the paper.

\section{System Model}
\label{sec:system-model}
Consider a ground control station (GCS), a digital twin, and \(N\) rotary-wing UAVs connected through a multi-hop low-altitude swarm network whose bidirectional topology varies across control cycles \cite{bekmezciFlyingAdHocNetworks2013}.
Each closed loop \(i\in\{1,\ldots,N\}\) couples one uplink telemetry flow with one downlink control flow, as illustrated in Fig.~\ref{fig:system-model-overview}.

\begin{figure}[t]
   \centering
   \includegraphics[width=\columnwidth, trim=10 10 10 10, clip]{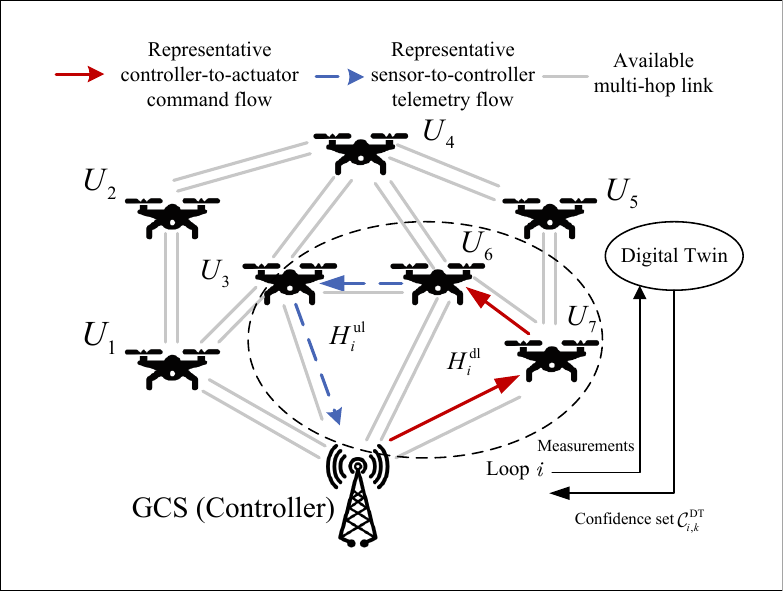}
   \caption{
      DTN-assisted multi-hop low-altitude UAV swarm network with the bidirectional loop for UAV~\(i\).
      Red solid arrows denote the downlink command flow along \(\mathcal{P}_i^{\mathrm{dl}}\!\left(k\right)\).
      Blue dashed arrows denote the uplink telemetry flow along \(\mathcal{P}_i^{\mathrm{ul}}\!\left(k\right)\).
      The two directional paths are generally asymmetric.
      Gray lines denote candidate swarm-network links whose active subset changes across control cycles.
      The digital twin receives measurements from the closed-loop swarm and provides the confidence set \(\mathcal{C}_{i,k}^{\mathrm{DT}}\) used for certificate construction.
   }
   \label{fig:system-model-overview}
\end{figure}

\subsection{Closed-Loop Dynamics and Admissible Domain}

Each UAV obeys discrete-time single-integrator kinematics
\begin{equation}
   \vct{p}_i(k+1) = \mtx{A}\vct{p}_i(k) + \mtx{B}\vct{u}_i(k) + \vct{w}_i(k),
   \label{eq:plant-dynamics}
\end{equation}
where \(\vct{p}_i(k)\in\mathbb{R}^3\) is the position, \(\vct{u}_i(k)\in\mathbb{R}^3\) is the velocity command, \(\mtx{A}=\mtx{I}_3\), \(\mtx{B}=T_s\mtx{I}_3\), and \(\vct{w}_i(k)\sim\mathcal{N}(\vct{0},\bmtx{\Sigma}_{\mathrm{proc}})\) is the process noise.
The sampling period \(T_s\) separates two time scales: the onboard flight controller stabilizes attitude and velocity at 50--200\,Hz, while the GCS issues velocity commands at the position-loop rate \(1/T_s\).
Because the inner-loop settling time is well below \(T_s\), the velocity response appears instantaneous at the position-loop time scale.
The position measurement obtained via uplink telemetry is \(\vct{p}_i^{\mathrm{meas}}(k)=\vct{p}_i(k)+\vct{v}_i(k)\), with \(\vct{v}_i(k)\sim\mathcal{N}(\vct{0},\bmtx{\Sigma}_{\mathrm{meas}})\).
Without a fresh command update, the tracking error grows as \(\vct{e}_i(k+1)=\vct{e}_i(k)+\vct{w}_i(k)\), since the controller propagates its estimate via the plant model.

Let \(\bmtx{\Sigma}^{\mathrm{p}}_i(k)\triangleq\mathbb{E}[\tilde{\vct{p}}_i(k)\tilde{\vct{p}}_i(k)^{\mathsf{T}}]\) denote the controller-side position-error covariance, where \(\tilde{\vct{p}}_i(k)=\hat{\vct{p}}_i(k)-\vct{p}_i(k)\).
With \(\mtx{A}=\mtx{I}_3\), the covariance update reduces to \(\bmtx{\Sigma}^{\mathrm{p}}_i(k+1)=\bmtx{\Sigma}^{\mathrm{p}}_i(k)+\bmtx{\Sigma}_{\mathrm{proc}}\).
The velocity command is generated by the LQR law \cite{andersonOptimalControlLinearQuadratic1989}
\begin{equation}
   \vct{u}_i(k)=-\mtx{K}\,\vct{e}_i(k),
   \label{eq:lqr-law}
\end{equation}
where \(\mtx{K}\in\mathbb{R}^{3\times 3}\) is the stabilizing feedback gain and \(\vct{e}_i(k)=\hat{\vct{p}}_i(k)-\vct{p}_i^{\mathrm{ref}}(k)\) is the controller-side tracking error with respect to the reference trajectory \(\vct{p}_i^{\mathrm{ref}}(k)\in\mathbb{R}^3\).
The controller propagates the most recently received uplink position sample through the plant model to obtain the estimate \(\hat{\vct{p}}_i(k)\).
If an uplink transmission succeeds within the one-way delay bound \(D_i^{\mathrm{ul}}\), the received sample is \(h_i^{\mathrm{ul}}=\lceil D_i^{\mathrm{ul}}/T_s\rceil\) cycles old; the controller propagates this delayed sample forward by \(h_i^{\mathrm{ul}}\) steps before computing the command.

The admissible domain is constructed from a discrete-time Lyapunov--Krasovskii functional (LKF) on the augmented tracking error:
\[
   V(\vct{z})=\vct{z}^{\mathsf{T}}\bmtx{X}\vct{z},
\]
where \(\vct{z}\in\mathbb{R}^{3(H+1)}\) stacks the current tracking error \(\vct{e}_i(k)\in\mathbb{R}^{3}\) and its \(H\) delayed copies, \(H\) is the offline delay-buffer depth covering the maximum admissible one-way delay, and \(\bmtx{X}\succ0\) is obtained offline from the LMI feasibility problem in \eqref{eq:lmi-certificate}.
The stabilizing feedback gain \(\mtx{K}\) follows from the discrete-time LQR solution \cite{andersonOptimalControlLinearQuadratic1989} and is distinct from the LKF matrix \(\bmtx{X}\).
\begin{definition}[Admissible Lyapunov Domain]
   The offline controller gain, safety-geometry budget, and estimation-noise floor define the rectangular admissible domain
   \begin{equation}
      \begin{aligned}
         \mathcal{Z}^{\mathrm{safe}}
         =
         \left\{
         \begin{array}{l}
            (V(\vct{z}),\tr(\bmtx{\Sigma}^{\mathrm{p}})):
            0\le V(\vct{z})\le V_{\max}, \\[2pt]
            \sigma_{\min}^{\mathrm{p}}\le\tr(\bmtx{\Sigma}^{\mathrm{p}})\le\sigma_{\max}^{\mathrm{p}}
         \end{array}
         \right\}.
      \end{aligned}
      \label{eq:safe-certification-domain}
   \end{equation}
   Here, \(V_{\max}\) is the maximum allowable tracking-error energy, set by the per-UAV safety corridor radius under the formation geometry budget.
   \(\sigma_{\max}^{\mathrm{p}}\) is the largest tolerable estimation uncertainty; beyond this threshold the LQR gain cannot maintain the required tracking performance.
   \(\sigma_{\min}^{\mathrm{p}}=\tr(\bmtx{\Sigma}_{\mathrm{meas}})\) is the measurement-noise floor, attained only after a zero-delay uplink refresh.
\end{definition}
The lower bound \(\tr(\bmtx{\Sigma}^{\mathrm{p}})\ge\sigma_{\min}^{\mathrm{p}}\) ensures that \(\tr(\bmtx{\Sigma}_{\mathrm{proc}})+\tr(\bmtx{\Sigma}^{\mathrm{p}})-\tr(\bmtx{\Sigma}_{\mathrm{meas}})\ge0\) for every \(\vct{\zeta}\in\mathcal{Z}^{\mathrm{safe}}\), thereby excluding states where a missed uplink would imply a tighter error bound than a zero-delay measurement.
Let \(\vct{\zeta}\triangleq(V(\vct{z}),\tr(\bmtx{\Sigma}^{\mathrm{p}}))\) denote the augmented Lyapunov coordinate; when loop indices are needed, \(\vct{\zeta}_i(k)=(V(\vct{z}_i(k)),\tr(\bmtx{\Sigma}^{\mathrm{p}}_i(k)))\).
The cycle-wise position-tracking error is
\begin{equation}
   \mathrm{RMSE}(k)=N^{-1/2}\|\mtx{P}(k)-\mtx{P}^{\mathrm{ref}}(k)\|_{\mathrm{F}},
   \label{eq:rmse}
\end{equation}
where \(\mtx{P}(k)=[\vct{p}_1(k),\ldots,\vct{p}_N(k)]\), \(\mtx{P}^{\mathrm{ref}}(k)=[\vct{p}_1^{\mathrm{ref}}(k),\ldots,\vct{p}_N^{\mathrm{ref}}(k)]\), and \(\|\cdot\|_{\mathrm{F}}\) is the Frobenius norm.

\subsection{Digital Twin and TDMA Action Model}

The network uses a TDMA frame with \(S\) slots per control cycle; one slot is the minimum scheduling unit.
Let \(T_{\mathrm{slot}}\) and \(T_{\mathrm{guard}}\) denote the slot duration and guard interval, and define \(\tau_{\mathrm{slot}}\triangleq T_{\mathrm{slot}}+T_{\mathrm{guard}}\), subject to \(S\tau_{\mathrm{slot}}\le T_s\).
Each hop supports acknowledgments and an action-dependent retransmission limit.
Link-layer ACK delivery uses the same retransmission mechanism as the forward transmission; the controller therefore infers the exact command-delivery delay after successful telemetry.
Under block fading, a transmission with coded blocklength \(n\) occupies \(\varsigma(n)=\lceil n/\lfloor B_{\mathrm{eff}}T_{\mathrm{slot}}\rfloor\rceil\) consecutive TDMA slots, where \(B_{\mathrm{eff}}\) is the effective coded-symbol rate.
For directed link \(e=(a,b)\), the received SNR during the block-fading slot is \(\gamma_e(k)\).
The dB-scale SNR follows from the link budget:
\[
   \begin{aligned}
      \gamma_e^{\mathrm{dB}}(k)
       & =
      P_{\mathrm{tx,dBm}}-L_e(k)+X_e(k)-N_{\mathrm{dBm}},
      \\
      N_{\mathrm{dBm}}
       & =
      -174+10\log_{10}B_{\mathrm{ch}}+F_{\mathrm{NF}},
   \end{aligned}
\]
where \(P_{\mathrm{tx,dBm}}\) is the transmit power, \(L_e(k)\) is the large-scale path loss, \(X_e(k)\) is the block-fading gain, \(N_{\mathrm{dBm}}\) is the receiver noise power, \(B_{\mathrm{ch}}\) is the channel bandwidth, and \(F_{\mathrm{NF}}\) is the receiver noise figure, all in their respective standard units.

The digital twin maintains the cyber-physical state for resource-allocation decisions \cite{wuDigitalTwinNetworks2021,khanDigitalTwinWireless2022}.
At the beginning of cycle \(k\), the twin state \(\mathcal{T}_k\) is updated from the mission plan, the latest accepted UAV telemetry, the controller-side propagated estimate \(\hat{\vct{p}}_i(k)\), the reference trajectory \(\vct{p}_i^{\mathrm{ref}}(k)\), available map or building data, and recent link-quality measurements.
The twin first predicts a short-horizon geometry state containing UAV positions, relay candidates, and obstacle-induced line-of-sight or non-line-of-sight link classes.
It then evaluates the link budget and fading for each directed candidate link, producing a predicted SNR \(\hat{\gamma}_e(k)\) and a predicted active-link indicator.
Route feasibility is evaluated on the predicted topology; the candidate UL/DL route sets are outputs of the twin state.
Let \(\mathcal{E}(k)\) denote the directed candidate-link set at cycle \(k\).
For loop \(i\), the twin provides the prediction interface
\begin{equation}
   \mathcal{T}_k
   \longmapsto
   \left(
   \{\hat{\gamma}_e(k)\}_{e\in\mathcal{E}(k)},\;
   \{\hat{a}_e(k)\}_{e\in\mathcal{E}(k)},\;
   \hat{\mathcal{P}}_i^{\mathrm{ul}}(k),\;
   \hat{\mathcal{P}}_i^{\mathrm{dl}}(k)
   \right),
   \label{eq:dt-state-map}
\end{equation}
where \(\hat{a}_e(k)\in\{0,1\}\) indicates predicted link availability after range, obstruction, and route-admission tests; the set \(\{e\in\mathcal{E}(k):\hat{a}_e(k)=1\}\) is the predicted active topology.
Let \(\vct{\omega}_i^{\mathrm{real}}(k)\) collect the realized channel qualities, link-availability indicators, route-feasibility state, and controller-side state summary corresponding to~\eqref{eq:dt-state-map}.
\begin{definition}[Digital-Twin Coverage]
   The twin provides a confidence set \(\mathcal{C}_{i,k}^{\mathrm{DT}}\) with budget \(\beta\in[0,1)\) satisfying
   \begin{equation}
      \Pr\!\left\{
      \vct{\omega}_i^{\mathrm{real}}(k)
      \in
      \mathcal{C}_{i,k}^{\mathrm{DT}}
      \;\middle|\;
      \mathcal{T}_k
      \right\}
      \ge
      1-\beta.
      \label{eq:dt-confidence-set}
   \end{equation}
   The set \(\mathcal{C}_{i,k}^{\mathrm{DT}}\) covers joint deviations of the SNRs, link-availability indicators, feasible route set, and controller-side state summary from their predicted values.
\end{definition}
Such a confidence set may be formed by residual calibration or conservative uncertainty bounds \cite{shaferTutorialConformal2008,benTalRobustConvex1998}.

Time-varying link opportunities define the finite action set available to the scheduler.
For direction \(\chi\in\{\mathrm{ul},\mathrm{dl}\}\), the routing layer provides a candidate path set \(\hat{\mathcal{P}}_i^\chi(k)\).
The finite transmission-configuration set \(\mathcal{X}\) enumerates retransmission limits and admissible coded blocklengths \(n\); the slot cost \(\varsigma(n)\) follows from the blocklength.
From \(\mathcal{T}_k\), the route sets, and the slot-cost classes induced by \(\mathcal{X}\), the scheduler generates a finite candidate slot-allocation set \(\mathcal{S}_i(k)\) on the TDMA grid.
Each element \(\sigma_i\in\mathcal{S}_i(k)\) is an ordered UL/DL service pattern that records the slot indices assigned to telemetry hops, command hops, relay forwarding, and reserved retransmission attempts within one control cycle.
Only patterns satisfying hop ordering, route-availability windows, guard intervals, per-hop slot occupation, and the cycle budget \(S\) are retained.
These patterns are certificate candidates rather than absolute frame schedules; their ordered slot positions determine the one-way delay coordinates \(D_i^\chi(\pi_i)\) and TDMA opportunity spacing \(\tau_i^{\mathrm{opp}}(\pi_i)\).
\begin{definition}[Feasible Communication Action Set]
   For loop \(i\), the feasible communication-action set is
   \begin{equation}
      \Pi_i(k)=
      \bigl(
      \hat{\mathcal{P}}_i^{\mathrm{ul}}(k)
      \times
      \hat{\mathcal{P}}_i^{\mathrm{dl}}(k)
      \bigr)
      \times
      \mathcal{S}_i(k)
      \times
      \mathcal{X}.
      \label{eq:feasible-action-set}
   \end{equation}
   At cycle \(k\), loop \(i\) selects one communication action \(\pi_i\in\Pi_i(k)\), which jointly specifies UL/DL routes, TDMA slot allocation, retransmission limits, and blocklengths.
   Let \(s_i(\pi_i)\in\mathbb{Z}_{+}\) denote the total number of TDMA slots occupied by \(\pi_i\), including reserved retransmission and relay-forwarding slots.
\end{definition}
The swarm network uses orthogonal collision-free TDMA without spatial reuse.
Any action subset satisfying \(\sum_i s_i(\pi_i)\le S\) can be packed into one frame by deterministic contiguous slot assignment, preserving each action's internal UL/DL order, route-hop order, slot span, opportunity spacing, and delay coordinates.
This orthogonalization excludes co-channel interference at the cost of spectral efficiency.
Each closed-loop interaction opportunity under \(\pi_i\) yields one of four outcomes: bidirectional delivery, uplink-only delivery, downlink-only delivery, or simultaneous failure.

\subsection{Zero-Allocation Continuation and Local Safe Mode}

If no network action is selected for loop \(i\) at cycle \(k\), the loop first attempts nominal controller-estimator continuation without bidirectional service.
The deterministic one-step hold map \(\mathcal{H}_i:\vct{\zeta}_i(k)\mapsto\vct{\zeta}_i(k+1)\) captures this zero-slot evolution:
\begin{align}
   V(\vct{z}_i(k+1))                      & = \alpha^{\mathrm{hold}}V(\vct{z}_i(k)), \label{eq:hold-v}                                       \\
   \tr(\bmtx{\Sigma}^{\mathrm{p}}_i(k+1)) & = \tr(\bmtx{\Sigma}^{\mathrm{p}}_i(k))+\tr(\bmtx{\Sigma}_{\mathrm{proc}}), \label{eq:hold-sigma}
\end{align}
where \(\alpha^{\mathrm{hold}}\ge 1\) is the open-loop growth factor of the lifted state under no communication, computed offline from the LKF at the maximum admissible delay.
Equation~\eqref{eq:hold-sigma} follows directly from \(\mtx{A}=\mtx{I}_3\).
The zero-slot continuation is \emph{hold-admissible} when \(\mathcal{H}_i(\vct{\zeta}_i(k))\in\mathcal{Z}^{\mathrm{safe}}\).

When the zero-slot continuation is not hold-admissible, the loop enters \emph{local safe mode}, modeled as a deterministic controlled-invariant subsystem with local control law \(\kappa_i^{\mathrm{safe}}\) and transition map \(f_i^{\mathrm{safe}}\) satisfying
\begin{equation}
   \vct{\zeta}_i(k)\in\mathcal{Z}^{\mathrm{safe}}
   \;\Rightarrow\;
   f_i^{\mathrm{safe}}(\vct{\zeta}_i(k))\in\mathcal{Z}^{\mathrm{safe}}.
   \label{eq:local-safe-invariance}
\end{equation}
The one-cycle contribution of this subsystem to the allocation objective is the deterministic Lyapunov drift
\begin{equation}
   J_i^{\mathrm{safe}}(\vct{\zeta}_i(k))
   =
   W(f_i^{\mathrm{safe}}(\vct{\zeta}_i(k)))
   -
   W(\vct{\zeta}_i(k)),
   \label{eq:local-safe-drift}
\end{equation}
which is independent of the stochastic QoS certificate and enters online allocation only through the baseline~\eqref{eq:allocation-baseline}.
While local safe mode is active, a successful bidirectional exchange serves as a synchronization event that resets the bidirectional-failure counter to \(c_i=0\); otherwise, the controlled-invariant mode persists.
The synchronization event does not claim a stochastic drift certificate for the current cycle.
Together, zero-slot nominal continuation and local safe mode define the behavior of loops that receive no certified network action.

\section{Shared QoS Certificate and Closed-Loop Admission}
\label{sec:certificate-admission}
Unless otherwise stated, the derivation fixes loop \(i\) at control cycle \(k\) and suppresses these indices.
Thus, \(T_k^{\mathrm{real}}(\pi;H_T)\), \(\tau^{\mathrm{opp}}(\pi)\), \(r(\pi)\), \(T^{\mathrm{cert}}(\pi)\), and \(\vct{q}(\pi)\) denote the fixed-loop forms of \(T_{i,k}^{\mathrm{real}}(\pi_i;H_T)\), \(\tau_i^{\mathrm{opp}}(\pi_i)\), \(r_i(\pi_i)\), \(T^{\mathrm{cert}}_i(\pi_i)\), and \(\vct{q}_i(\pi_i;k)\), respectively.
Joint allocation across loops is stated explicitly in Section~\ref{sec:frontier-allocation}.

\subsection{Shared QoS Certificate}

Discrete multi-hop actions require a common descriptor that captures the delay, interaction, and delivery properties relevant to closed-loop support.
Conventional wireless QoS centers on delay, throughput, reliability, and jitter \cite{dapengwuEffectiveCapacityWireless2003}, whereas wireless networked control depends on directional delays, bidirectional interaction intervals, and loop dynamics rather than isolated packet delivery \cite{heemelsNetworkedControlSystems2010}.
The certificate records end-to-end one-way delay bounds, directional service-success guarantees, and a certified upper bound on the interval between consecutive successful bidirectional interactions \cite{huangOptimalDownlinkUplink2020}.
\begin{definition}[Shared QoS Certificate]
   A feasible action \(\pi\in\Pi\) first induces a raw supply certificate \(\vct{q}^{\mathrm{raw}}(\pi)\).
   After conservative rounding and offline admission, the online allocation rule selects its representative \(\vct{q}(\pi)\in\mathcal{Q}\).
   Both forms share the coordinate representation
   \begin{equation}
      \vct{q}=
      \left(
      D^{\mathrm{ul}},\;
      D^{\mathrm{dl}},\;
      T^{\mathrm{cert}},\;
      \rho^{\mathrm{ul}},\;
      \rho^{\mathrm{dl}}
      \right),
   \end{equation}
   where \(D^{\mathrm{ul}}\) and \(D^{\mathrm{dl}}\) are end-to-end one-way delay bounds, \(T^{\mathrm{cert}}\) is a certified upper bound on the interval between consecutive successful bidirectional interactions, and \(\rho^{\mathrm{ul}},\rho^{\mathrm{dl}}\) are per-cycle lower bounds on directional service success.
\end{definition}
For offline MJLS certification, the route-slot pattern associated with \(\vct{q}\) induces the timing triple \((h^{\mathrm{ul}},h^{\mathrm{dl}},g)\) with \(h^\chi=\lceil D^\chi/T_s\rceil\) and \(g=T^{\mathrm{cert}}/T_s\).

Let \(\mathcal{Q}\) denote the finite online certificate library after offline admission.
For \(\vct{q},\tilde{\vct{q}}\in\mathcal{Q}\), define \(\vct{q}\succeq\tilde{\vct{q}}\) iff \(D^\chi\le \tilde{D}^\chi\), \(T^{\mathrm{cert}}\le \tilde{T}^{\mathrm{cert}}\), and \(\rho^\chi\ge \tilde{\rho}^\chi\) for \(\chi\in\{\mathrm{ul},\mathrm{dl}\}\).
The coordinate \(T^{\mathrm{cert}}\) is an interaction-level quantity induced by the route-slot pattern and its bidirectional failure-tolerance guarantee.

\subsection{Digital-Twin-Assisted Certificate Generation}

For each loop and candidate action, the twin maps predicted topology, channel state, and closed-loop state into the five-dimensional certificate:
\begin{equation}
   (\mathcal{T}_k,\pi_i)
   \longmapsto
   \vct{q}_i(\pi_i;k)
   =
   \left(
   D_i^{\mathrm{ul}},\;
   D_i^{\mathrm{dl}},\;
   T^{\mathrm{cert}}_i,\;
   \rho_i^{\mathrm{ul}},\;
   \rho_i^{\mathrm{dl}}
   \right).
   \label{eq:dt-certificate-map}
\end{equation}
The predicted channel qualities \(\hat{\gamma}_e(k)\), link availability, candidate paths, slot patterns, retransmission limits, blocklengths, and controller-side summaries parameterize \(\vct{q}_i(\pi_i;k)\).
Using the confidence set \(\mathcal{C}_{i,k}^{\mathrm{DT}}\) from~\eqref{eq:dt-confidence-set}, the certificate construction evaluates each action at the worst-case delay, deadline-level service failure, and failure-persistence values within the set.
The resulting certificate is conservative whenever \(\vct{\omega}_i^{\mathrm{real}}(k)\in\mathcal{C}_{i,k}^{\mathrm{DT}}\).
Let \(\vct{q}_i^{\mathrm{real}}(\pi_i;k)\) denote the certificate obtained by applying the same coordinate map to the realized channel, topology, route-feasibility, and controller-side state at cycle~\(k\).
\begin{remark}[Twin-to-Certificate Soundness]
   \label{rem:twin-certificate-soundness}
   Suppose \(\mathcal{C}_{i,k}^{\mathrm{DT}}\) satisfies~\eqref{eq:dt-confidence-set} and \(\vct{q}_i(\pi_i;k)\) is constructed by worst-case evaluation over \(\mathcal{C}_{i,k}^{\mathrm{DT}}\).
   By the definition of the confidence set and the certificate order~\(\succeq\), the realized supply certificate satisfies \(\vct{q}_i^{\mathrm{real}}(\pi_i;k)\succeq\vct{q}_i(\pi_i;k)\) with probability at least \(1-\beta\) conditional on~\(\mathcal{T}_k\).
\end{remark}

For a feasible action \(\pi\in\Pi\), let \(n_{\mathrm{re}}^\chi(\pi)\), \(n^\chi(\pi)\), \(L^\chi(\pi)\), and \(\ell^\chi(\pi)\) denote its direction-\(\chi\) retransmission limit, coded blocklength, hop count, and scheduled slot-clock span, respectively.
The span \(\ell^\chi(\pi)\) counts scheduled TDMA slots along the selected route; \(s(\pi)\) denotes the total TDMA slots occupied by \(\pi\), including reserved retransmission and relay-forwarding slots.
Let \(x_i(k)\) denote the service-conditioning state formed from \(\mathcal{T}_k\), the controller-side state \(\vct{\zeta}_i(k)\), and the calibrated physical uncertainty set.
For each direction \(\chi\in\{\mathrm{ul},\mathrm{dl}\}\), the digital-twin-calibrated service model provides an upper bound on the probability of missing the deadline:
\begin{equation}
   q_{\mathrm{service}}^\chi(\pi,x_i(k))
   =
   \sup_{\vct{\omega}\in\mathcal{C}_{i,k}^{\mathrm{DT}}}
   \Pr\!\left\{\mathcal{F}_{\mathrm{svc}}^\chi \mid \pi,x_i(k),\vct{\omega}\right\},
   \label{eq:dt-service-failure}
\end{equation}
where \(\mathcal{F}_{\mathrm{svc}}^\chi\) denotes the event that direction-\(\chi\) service misses its deadline.
The directional service-success coordinate is
\begin{equation}
   \rho^\chi(\pi)=1-q_{\mathrm{service}}^\chi(\pi,x_i(k)).
   \label{eq:directional-delivery}
\end{equation}
If no calibrated bound is available for \((\pi,x_i(k))\), the action is excluded from \(\mathcal{Q}_i(k)\).
The directional delay coordinate is
\begin{equation}
   D^\chi(\pi)=\ell^\chi(\pi)\,\tau_{\mathrm{slot}}.
   \label{eq:directional-delay}
\end{equation}
Let \(\tau^{\mathrm{opp}}(\pi)\) denote the base duration of one bidirectional scheduling opportunity induced by the slot pattern of \(\pi\), and let \(a^\chi_\pi\) and \(b^\chi_\pi\) be the first and last TDMA slot indices occupied by direction~\(\chi\).
Then
\begin{equation}
   \tau^{\mathrm{opp}}(\pi)
   =
   \bigl(
   \max\{b^{\mathrm{ul}}_\pi,b^{\mathrm{dl}}_\pi\}
   -
   \min\{a^{\mathrm{ul}}_\pi,a^{\mathrm{dl}}_\pi\}
   +1
   \bigr)\tau_{\mathrm{slot}}.
   \label{eq:opportunity-spacing}
\end{equation}
The GCS control architecture admits at most one effective closed-loop interaction per control cycle; the interaction cadence therefore equals the sampling period \(T_s\).
\begin{remark}
   The framework extends to multiple bidirectional opportunities per cycle by replacing \(T_s\) with a per-opportunity cadence and scaling the single-opportunity contraction envelope accordingly.
\end{remark}

\subsection{Certified Interaction-Bound Coordinate}

Let \(H_T\) denote a finite horizon of bidirectional opportunities, chosen offline to span several certified interaction intervals.
Let \(\{t_m\}_{m=1}^{N_{H_T}}\) denote the strictly increasing sequence of successful bidirectional interaction times within the next \(H_T\) opportunities under action~\(\pi\).
The adjacent successful-interaction index set is \(\mathcal{M}(H_T)=\{m\in\mathbb{Z}_{+}:m<N_{H_T}\}\).
If \(\mathcal{M}(H_T)\neq\varnothing\), the realized rolling-window maximum interaction interval is
\begin{equation}
   T_k^{\mathrm{real}}(\pi;H_T)
   \triangleq
   \max_{m\in\mathcal{M}(H_T)}
   (t_{m+1}-t_{m});
   \label{eq:interaction-interval}
\end{equation}
otherwise, \(T_k^{\mathrm{real}}(\pi;H_T)=+\infty\).

Let \(q_{\mathrm{service}}^{\mathrm{bi}}(\pi,x_i(k))\) denote the digital-twin upper bound on same-opportunity bidirectional deadline-level service failure.
The corresponding conservative lower bound on bidirectional success is
\begin{equation}
   \underline{\rho}^{\mathrm{bi}}(\pi)
   =
   1-
   q_{\mathrm{service}}^{\mathrm{bi}}(\pi,x_i(k)).
   \label{eq:dt-bidirectional-service-success}
\end{equation}

Temporal correlation across bidirectional opportunities is modeled by a two-state failure chain following the Gilbert--Elliott burst-error abstraction \cite{gilbertCapacityBurstNoise1960,elliottEstimatesErrorRates1963}.
Let \(Y_m(\pi)\in\{0,1\}\) denote the bidirectional outcome at the \(m\)th opportunity under action~\(\pi\), where \(0\) denotes success and \(1\) denotes failure.
The transition matrix is
\begin{equation}
   \mtx{P}^{\mathrm{F}}(\pi)
   =
   \begin{bmatrix}
      p_{00}(\pi) & p_{01}(\pi) \\
      p_{10}(\pi) & p_{11}(\pi)
   \end{bmatrix},
   \label{eq:failure-transition}
\end{equation}
where \(p_{00}(\pi)=\Pr\{Y_{m+1}=0\mid Y_m=0\}\), \(p_{11}(\pi)=\Pr\{Y_{m+1}=1\mid Y_m=1\}\), \(p_{01}=1-p_{00}\), and \(p_{10}=1-p_{11}\).

The digital twin supplies two conservative failure-chain parameters by evaluating service-outcome sequences over \(\mathcal{C}_{i,k}^{\mathrm{DT}}\):
an upper bound \(\mu_1(\pi)\) on the probability that the first opportunity is a bidirectional failure, and an upper bound \(p_{11}(\pi)\) on the conditional persistence probability \(\Pr\{Y_{m+1}=1\mid Y_m=1\}\).
By default, \(\mu_1(\pi)=1-\underline{\rho}^{\mathrm{bi}}(\pi)\); when the digital twin resolves the initial failure-chain state from the interaction history, \(\mu_1(\pi)\) can be tightened to the state-conditioned first-opportunity failure probability.
The probability of \(L\) consecutive bidirectional failures is upper-bounded by
\begin{equation}
   \Pr\{Y_1=1,\ldots,Y_L=1\}
   \le
   \mu_1(\pi)\,p_{11}(\pi)^{L-1},
   \qquad L\ge 1.
   \label{eq:markov-failure-tail}
\end{equation}
Let \(\delta_T\in(0,1)\) denote the per-cycle interaction-certificate violation budget, assigned to the finite-horizon bidirectional-failure tail.
For a horizon of \(H_T\) bidirectional opportunities, the consecutive-failure tolerance is
\begin{equation}
   L^{\mathrm{F}}(\pi)
   \triangleq
   \min
   \left\{
   L\in\mathbb{Z}_{+}:
   H_T\,\mu_1(\pi)\,p_{11}(\pi)^{L-1}
   \le
   \delta_T
   \right\}.
   \label{eq:markov-failure-tolerance}
\end{equation}
If the set in~\eqref{eq:markov-failure-tolerance} is empty, the action does not admit a finite interaction certificate.
Let \(g(\pi)\triangleq L^{\mathrm{F}}(\pi)\); at most \(g(\pi)-1\) consecutive failed opportunities can occur between two successful bidirectional interactions, and a block of \(g(\pi)\) consecutive failures constitutes the certified violation event.
A discussion of multi-hop applicability is provided in Appendix~\ref{sec:appendix-failure-chain}.
The certified interaction-bound coordinate is
\begin{equation}
   T^{\mathrm{cert}}(\pi)
   \triangleq
   g(\pi)\,T_s.
   \label{eq:tuple-interaction-cert}
\end{equation}
Each feasible action with finite \(\tau^{\mathrm{opp}}(\pi)\) and finite \(L^{\mathrm{F}}(\pi)\) induces \(\vct{q}(\pi)=(D^{\mathrm{ul}}(\pi),D^{\mathrm{dl}}(\pi),T^{\mathrm{cert}}(\pi),\rho^{\mathrm{ul}}(\pi),\rho^{\mathrm{dl}}(\pi))\).
\begin{proposition}[Rolling Finite-Horizon Interaction Reliability]
   \label{prop:finite-horizon-interaction-reliability}
   At control cycle \(k\), condition on the digital-twin state \(\mathcal{T}_k\) and on an action \(\pi\) whose conservative failure-chain parameters satisfy~\eqref{eq:markov-failure-tolerance}.
   Over the next \(H_T\) bidirectional opportunities induced by~\(\pi\),
   \begin{equation}
      \Pr\!\left\{
      T_k^{\mathrm{real}}(\pi;H_T)
      \le
      g(\pi)\,T_s
      \;\middle|\;
      \mathcal{T}_k
      \right\}
      \ge
      1-\delta_T.
   \end{equation}
\end{proposition}
\begin{proof}
   The horizon contains at most \(H_T\) starting positions for a block of \(g(\pi)=L^{\mathrm{F}}(\pi)\) consecutive bidirectional failures.
   Applying the union bound to~\eqref{eq:markov-failure-tail} yields a violation probability no larger than \(H_T\,\mu_1(\pi)\,p_{11}(\pi)^{g(\pi)-1}\), which is at most \(\delta_T\) by~\eqref{eq:markov-failure-tolerance}.
   When no such block occurs, at most \(g(\pi)-1\) failed opportunities separate two successful interactions, so the realized interval is bounded by \(g(\pi)\,T_s\).
\end{proof}

\subsection{Stochastic Drift Certificate}

The certificate order supports allocation only if stronger certificates imply no larger closed-loop drift.
Define the augmented Lyapunov function
\begin{equation}
   W(\vct{z},\bmtx{\Sigma}^{\mathrm{p}})
   \triangleq
   V(\vct{z})
   +
   \lambda_\Sigma\cdot\tr(\bmtx{\Sigma}^{\mathrm{p}}),
   \label{eq:augmented-lyapunov}
\end{equation}
where \(\vct{z}\) stacks \(\vct{e}_i\) and its delayed samples, \(V(\vct{z})\) is the LKF defined in Section~\ref{sec:system-model}, \(\bmtx{\Sigma}^{\mathrm{p}}\) is the controller-side position-error covariance, and \(\lambda_\Sigma>0\) weights estimation uncertainty relative to tracking-error energy.
Intuitively, \(W\) combines trajectory-tracking quality (\(V\)) with estimation confidence (\(\tr(\bmtx{\Sigma}^{\mathrm{p}})\)) into a single scalar that the allocation mechanism seeks to reduce.

For any certificate \(\vct{q}\in\mathcal{Q}\), let \(\bar{\alpha}(\vct{q},c)\) denote the certified one-cycle contraction envelope of \(V(\vct{z})\), where \(c=c(\vct{\zeta})\) is the incoming bidirectional-failure run counter.
Each control cycle contains exactly one bidirectional opportunity.
If the opportunity succeeds, the per-opportunity envelope \(\bar{\alpha}_{\mathrm{nf}}(\vct{q})\) applies; if it fails, the run extends from \(c\) to \(c+1\) and the run-length factor \(\alpha_{(c+1)}(\vct{q})\) applies.
The envelope takes the worse case:
\begin{equation}
   \bar{\alpha}(\vct{q},c)
   =
   \max\bigl\{
   \bar{\alpha}_{\mathrm{nf}}(\vct{q}),\;
   \alpha_{(c+1)}(\vct{q})
   \bigr\},
   \qquad c=0,\ldots,g-1.
   \label{eq:robust-alpha}
\end{equation}
The factors \(\bar{\alpha}_{\mathrm{nf}}\) and \(\alpha_{(c+1)}\) are computed offline from the cycle-level MJLS parameterized by \((h^{\mathrm{ul}},h^{\mathrm{dl}},g)\), where \(h^\chi=\lceil D^\chi/T_s\rceil\) and \(g=T^{\mathrm{cert}}/T_s\).
The MJLS uses the LKF matrix \(\bmtx{X}\succ0\) from the LMI feasibility problem and the Gilbert--Elliott failure-chain certificate from Proposition~\ref{prop:finite-horizon-interaction-reliability}.
Consecutive bidirectional failures are evaluated through run-length matrices \(\mtx{M}_{0,0}^{(j)}\), which differ from the repeated product \((\mtx{M}_{0,0}^{\mathrm{opp}})^j\); the certificate thus tracks how feedback staleness grows with the failure run.
Because \(\mtx{A}=\mtx{I}_3\), the lifted matrix stabilizes once the buffer saturates with open-loop entries; hence \(\alpha_{(j)}(\vct{q})=\alpha_{(H)}(\vct{q})\) for all \(j\ge H\), and only the first \(H\) run-length factors require offline computation.
The full construction is provided in Appendix~\ref{sec:appendix-contraction-envelope}.
All quantities in~\eqref{eq:robust-alpha} except the incoming counter \(c\) are computed offline.
The online admission test reads \(c\) from the controller-side interaction state and evaluates the scalar envelope \(\bar{\alpha}(\vct{q},c)\) for each candidate certificate.

\begin{lemma}[Mode-Order Property]
   \label{lem:mode-order}
   For \(\mtx{A}=\mtx{I}_3\) and any \(\bmtx{X}\succ0\) satisfying~\eqref{eq:lmi-certificate}, the certified per-opportunity factors satisfy \(\alpha_{11}\le\alpha_{10}\le\alpha_{00}\) and \(\alpha_{11}\le\alpha_{01}\le\alpha_{00}\).
   The two one-sided factors \(\alpha_{10}\) and \(\alpha_{01}\) need not be ordered with respect to each other.
\end{lemma}
The proof is given in Appendix~\ref{sec:appendix-mode-order}.
Since each degraded-feedback mode has a contraction factor no smaller than the bidirectional-success factor, \(\bar{\alpha}(\vct{q},c)\ge 1\), and \(\vct{q}\succeq\tilde{\vct{q}}\) implies \(\bar{\alpha}(\vct{q},c)\le\bar{\alpha}(\tilde{\vct{q}},c)\) for any admissible \(c\) (Proposition~\ref{prop:drift-monotonicity}, proved in Appendix~\ref{sec:appendix-drift-monotonicity}).

A successful uplink refresh resets the controller-side position-error covariance.
When the telemetry sample arrives with delay index \(h^{\mathrm{ul}}\), the controller propagates the delayed sample forward by \(h^{\mathrm{ul}}\) steps.
Since \(\mtx{A}=\mtx{I}_3\), each propagation step adds exactly \(\tr(\bmtx{\Sigma}_{\mathrm{proc}})\) without amplifying the prior estimation error.
The resulting covariance floor is
\(\bar{\Sigma}^{\mathrm{ul}}(\vct{q}) = \tr(\bmtx{\Sigma}_{\mathrm{meas}})+h^{\mathrm{ul}}\,\tr(\bmtx{\Sigma}_{\mathrm{proc}})\);
when \(h^{\mathrm{ul}}=0\), this equals the measurement-noise floor.
Without a fresh measurement, the controller propagates the covariance by adding \(\tr(\bmtx{\Sigma}_{\mathrm{proc}})\).
With at least one uplink success, the covariance resets to \(\bar{\Sigma}^{\mathrm{ul}}(\vct{q})\).

The drift bound combines these two mutually exclusive events.
Because the expected covariance change is affine in the uplink-success probability \(p\in[\rho^{\mathrm{ul}},1]\), the worst case is attained at the endpoint determined by the sign of the gradient.
Let \(\Delta_{\Sigma}(\vct{\zeta},\vct{q})\triangleq\bar{\Sigma}^{\mathrm{ul}}(\vct{q})-\tr(\bmtx{\Sigma}^{\mathrm{p}})-\tr(\bmtx{\Sigma}_{\mathrm{proc}})\).
Define
\begin{equation}
   \Psi_{\Sigma}(\vct{\zeta},\vct{q})
   \triangleq
   \begin{cases}
      \Delta_{\Sigma}+\tr(\bmtx{\Sigma}_{\mathrm{proc}}),
       & \Delta_{\Sigma}\ge 0, \\[4pt]
      \rho^{\mathrm{ul}}\Delta_{\Sigma}+\tr(\bmtx{\Sigma}_{\mathrm{proc}}),
       & \Delta_{\Sigma}<0.
   \end{cases}
   \label{eq:covariance-endpoint-envelope}
\end{equation}
Intuitively, when \(\Delta_{\Sigma}\ge0\), the uplink-refresh covariance floor exceeds the open-loop propagation and the worst case corresponds to \(p=1\); when \(\Delta_{\Sigma}<0\), the refresh is beneficial and the worst case corresponds to \(p=\rho^{\mathrm{ul}}\).
The expected one-cycle change of the augmented Lyapunov function on the certified interaction event satisfies \(\mathbb{E}[\Delta W \mid \vct{\zeta},\vct{q}]\le\Phi(\vct{\zeta},\vct{q})\), where
\begin{equation}
   \begin{aligned}
      \Phi(\vct{\zeta},\vct{q})
      \triangleq\  &
      \bigl(\bar{\alpha}(\vct{q},c(\vct{\zeta}))-1\bigr)V(\vct{z}) \\
                   & +
      \lambda_\Sigma\,\Psi_{\Sigma}(\vct{\zeta},\vct{q}).
   \end{aligned}
   \label{eq:drift-bound}
\end{equation}
The derivation is provided in Appendix~\ref{sec:appendix-drift-bound}.

When \(\Phi(\vct{\zeta},\vct{q})\le0\), the selected certificate yields the conditional drift inequality \(\mathbb{E}[W_{k+1}\mid\mathcal{F}_k,\mathcal{E}_T]\le W_k\) for states in \(\mathcal{Z}^{\mathrm{safe}}\), where \(\mathcal{E}_T\) is the certified interaction event.
If the policy selects safe certificates over a certified interval, the state remains in \(\mathcal{Z}^{\mathrm{safe}}\), and both the digital-twin coverage and interaction-certification events hold throughout, then \(W_k\) is a nonnegative supermartingale on that certified event.
Under this event-conditioned supermartingale, Doob's maximal inequality gives \(\Pr\{\sup_{0\le t\le k} W_t>\eta\mid\cap_{s=0}^{k}\mathcal{E}_{T,s}\}\le W_0/\eta\) for any finite horizon~\(k\) \cite{doCostaDiscreteTimeMarkov2005}.
The finite-horizon probability of the certified event is controlled separately by~\eqref{eq:rolling-certificate-budget}.
\begin{proposition}[Drift-Certificate Monotonicity]
   \label{prop:drift-monotonicity}
   For states with \(c(\vct{\zeta})<g(\vct{q})\), the augmented Lyapunov drift upper bound \(\Phi(\vct{\zeta},\vct{q})\) in~\eqref{eq:drift-bound} is nondecreasing in \(D^{\mathrm{ul}}\), \(D^{\mathrm{dl}}\), and \(T^{\mathrm{cert}}\), and nonincreasing in \(\rho^{\mathrm{ul}}\) and \(\rho^{\mathrm{dl}}\).
   Hence \(\vct{q}\succeq\tilde{\vct{q}}\) implies \(\Phi(\vct{\zeta},\vct{q})\le \Phi(\vct{\zeta},\tilde{\vct{q}})\).
\end{proposition}
The proof is provided in Appendix~\ref{sec:appendix-drift-monotonicity}.

\subsection{State-Conditioned Safe Certificate Test}

\begin{definition}[State-Conditioned Safe Certificate Set]
   Let \(\mathcal{Q}_i(k)\subseteq\mathcal{Q}\) denote the certificate representatives induced by loop~\(i\)'s candidate actions at cycle~\(k\).
   The safe certificate set for loop~\(i\) is
   \begin{equation}
      \mathcal{Q}^{\mathrm{safe}}_i(k)
      =
      \left\{
      \vct{q}\in\mathcal{Q}_i(k):
      \Phi_i(\vct{\zeta}_i(k),\vct{q})
      \le0
      \right\},
      \label{eq:safe-certificate-region}
   \end{equation}
   where \(\vct{\zeta}_i(k)\in\mathcal{Z}^{\mathrm{safe}}\) is the current augmented Lyapunov state.
   For any \(\vct{q}\in\mathcal{Q}^{\mathrm{safe}}_i(k)\), the augmented Lyapunov function has nonpositive conditional drift for the current cycle on the certified interaction event.
\end{definition}
State-conditioned admission reduces to the scalar check \(\Phi_i(\vct{\zeta}_i(k),\vct{q})\le0\), rather than a supremum over \(\mathcal{Z}^{\mathrm{safe}}\).
When \(\mathcal{Q}^{\mathrm{safe}}_i(k)\) is empty, network-assisted certification is unavailable for loop~\(i\) in cycle~\(k\), and the controller relies on the zero-allocation or local safe-mode path defined in Section~\ref{sec:system-model}.
The set \(\mathcal{Q}^{\mathrm{safe}}_i(k)\) is upper closed under \(\succeq\) by Proposition~\ref{prop:drift-monotonicity}.

\subsection{Probabilistic Envelope Admission}

\begin{proposition}[Probabilistic Envelope Admission]
   \label{prop:probabilistic-certificate-safety}
   Suppose \(\Phi_i(\vct{\zeta}_i(k),\vct{q}_i(\pi_i;k))\le0\), \(\mathcal{C}_{i,k}^{\mathrm{DT}}\) satisfies~\eqref{eq:dt-confidence-set}, and \(L^{\mathrm{F}}(\pi_i)\) satisfies~\eqref{eq:markov-failure-tolerance}.
   On the digital-twin coverage event \(\mathcal{E}_{\mathrm{DT}}\), the realized action satisfies both the state-conditioned stochastic drift certificate and the rolling finite-horizon interaction certificate with probability at least \(1-\delta_T\):
   \begin{equation}
      \begin{aligned}
         \Pr\bigl\{ &
         \Phi_i(\vct{\zeta}_i(k),\vct{q}_i^{\mathrm{real}}(\pi_i;k))
         \le 0,      \\
                    &
         T_{i,k}^{\mathrm{real}}(\pi_i;H_T)
         \le
         T^{\mathrm{cert}}_i(\pi_i)
         \;\big|\;
         \mathcal{E}_{\mathrm{DT}},\,
         \mathcal{T}_k
         \bigr\}
         \ge
         1-\delta_T.
      \end{aligned}
   \end{equation}
\end{proposition}
\begin{proof}
   Let \(\mathcal{E}_{\mathrm{DT}}\) be the event \(\vct{\omega}_i^{\mathrm{real}}(k)\in\mathcal{C}_{i,k}^{\mathrm{DT}}\), and let \(\mathcal{E}_{T}\) be the event \(T_{i,k}^{\mathrm{real}}(\pi_i;H_T)\le T^{\mathrm{cert}}_i(\pi_i)\).
   Remark~\ref{rem:twin-certificate-soundness} gives \(\vct{q}_i^{\mathrm{real}}(\pi_i;k)\succeq\vct{q}_i(\pi_i;k)\) on \(\mathcal{E}_{\mathrm{DT}}\).
   The digital twin constructs \(\mu_1(\pi_i)\) and \(p_{11}(\pi_i)\) as worst-case upper bounds over \(\mathcal{C}_{i,k}^{\mathrm{DT}}\); on \(\mathcal{E}_{\mathrm{DT}}\), the true parameters satisfy \(\mu_1^{\mathrm{real}}\le\mu_1(\pi_i)\) and \(p_{11}^{\mathrm{real}}\le p_{11}(\pi_i)\).
   Proposition~\ref{prop:finite-horizon-interaction-reliability} therefore applies conditionally on \(\mathcal{E}_{\mathrm{DT}}\), yielding \(\Pr\{\mathcal{E}_{T}\mid\mathcal{E}_{\mathrm{DT}},\mathcal{T}_k\}\ge1-\delta_T\).
   Since \(\Phi_i(\vct{\zeta}_i(k),\vct{q}_i(\pi_i;k))\le0\) and \(\vct{q}_i^{\mathrm{real}}(\pi_i;k)\succeq\vct{q}_i(\pi_i;k)\) on \(\mathcal{E}_{\mathrm{DT}}\), Proposition~\ref{prop:drift-monotonicity} gives \(\Phi_i(\vct{\zeta}_i(k),\vct{q}_i^{\mathrm{real}}(\pi_i;k))\le0\) on \(\mathcal{E}_{\mathrm{DT}}\).
   Hence both the state-conditioned drift certificate and the interaction certificate hold with conditional probability at least \(1-\delta_T\) given \(\mathcal{E}_{\mathrm{DT}}\) and \(\mathcal{T}_k\).
\end{proof}
For a finite set \(\mathcal{K}\) of certified control cycles, rolling certificates compose by explicit violation-budget accounting.
If cycle \(k\in\mathcal{K}\) uses interaction budget \(\delta_{T,k}\), then on the corresponding digital-twin coverage events,
\begin{equation}
   \Pr\!\left\{
   \bigcap_{k\in\mathcal{K}}
   \mathcal{E}_{T,k}
   \;\middle|\;
   \bigcap_{k\in\mathcal{K}}\mathcal{E}_{\mathrm{DT},k},\;
   \{\mathcal{T}_k\}_{k\in\mathcal{K}}
   \right\}
   \ge
   1-
   \sum_{k\in\mathcal{K}}
   \delta_{T,k}.
   \label{eq:rolling-certificate-budget}
\end{equation}
For a mission of known duration, the per-cycle budgets \(\delta_{T,k}\) are chosen so that the accumulated sum stays within the required confidence level.
Digital-twin prediction uncertainty is accounted for when constructing conservative \(\rho^\chi\), delay, and persistence coordinates over \(\mathcal{C}_{i,k}^{\mathrm{DT}}\), rather than by splitting \(\delta_T\).
The union-bound certificate calibrates \(L^{\mathrm{F}}\) offline.

\section{Certified Supply Frontier Allocation}
\label{sec:frontier-allocation}

Given the state-conditioned safe certificate set \(\mathcal{Q}^{\mathrm{safe}}_i(k)\), online allocation proceeds through certificate evaluation, safety filtering, frontier construction, and TDMA slot assignment.

\subsection{Certificate-Cost Dominance and Frontier Construction}

Let \(\Pi_i^{\mathrm{safe}}(k)\triangleq\{\pi\in\Pi_i(k):\Phi_i(\vct{\zeta}_i(k),\vct{q}_i(\pi;k))\le0\}\) be the safe-feasible action set after state-conditioned admission.
For actions \(\pi,\pi'\in\Pi_i^{\mathrm{safe}}(k)\), define the certificate-cost dominance relation \(\pi'\triangleright\pi\) if \(\vct{q}_i(\pi')\succeq\vct{q}_i(\pi)\) and \(s_i(\pi')\le s_i(\pi)\), with at least one strict inequality.
The certified supply frontier is
\begin{equation}
   \Pi_i^{\mathrm{fr}}(k) \triangleq \left\{ \pi\in\Pi_i^{\mathrm{safe}}(k) : \nexists\,\pi'\in\Pi_i^{\mathrm{safe}}(k),\ \pi'\triangleright \pi \right\}.
\end{equation}
The frontier removes candidates that consume more slots while providing weaker closed-loop support.
Frontier pruning preserves both safety admission and allocation ordering at the current state because the same drift bound \(\Phi\) defines \(\mathcal{Q}^{\mathrm{safe}}_i(k)\), \(\Pi_i^{\mathrm{safe}}(k)\), and \(U_i(\pi)\).
Points below the two-dimensional upper envelope (slot cost vs.\ Lyapunov-drift reduction) may still be valid when their full certificates are incomparable.

\begin{proposition}[Certificate Preservation and Frontier Sufficiency]
   \label{prop:frontier-sufficiency}
   Fix control cycle \(k\).
   Let \(\pi\in\Pi\) be a feasible action with certificate \(\vct{q}(\pi)\).
   If there exists a certificate \(\vct{q}^*\) such that \(\Phi(\vct{\zeta}(k),\vct{q}^*)\le0\) and \(\vct{q}(\pi)\succeq\vct{q}^*\), then \(\Phi(\vct{\zeta}(k),\vct{q}(\pi))\le0\).
   If \(\pi,\pi'\in\Pi_i^{\mathrm{safe}}(k)\) and \(\pi' \triangleright \pi\), then \(U_i(\pi')\ge U_i(\pi)\) and \(s_i(\pi')\le s_i(\pi)\).
   Therefore, every action excluded from \(\Pi_i^{\mathrm{fr}}(k)\) is unnecessary for optimal online allocation under the harm-reduction objective \(U_i(\pi)\) and shared budget~\(S\).
\end{proposition}
The proof is provided in Appendix~\ref{sec:appendix-frontier-sufficiency}.

\subsection{State-Conditioned Harm Reduction}

Let \(\pi_i^0\) denote the zero-slot nominal continuation action from Section~\ref{sec:system-model}, with \(s_i(\pi_i^0)=0\), and define \(\Pi_i^{+}(k)\triangleq\Pi_i(k)\cup\{\pi_i^0\}\).
Let \(J_i^0(\vct{\zeta}_i)\triangleq W(\mathcal{H}_i(\vct{\zeta}_i))-W(\vct{\zeta}_i)\) denote the one-cycle Lyapunov drift under the deterministic hold map~\(\mathcal{H}_i\).
The hold-map growth factor \(\alpha^{\mathrm{hold}}\) in~\eqref{eq:hold-v} is obtained from the bidirectional-failure transition of the LKF using the maximum admissible delay and interaction indices; it is not specific to any individual certificate.
Define \(J_i(\pi_i,\vct{\zeta}_i)=J_i^0(\vct{\zeta}_i)\) for \(\pi_i=\pi_i^0\), and \(J_i(\pi_i,\vct{\zeta}_i)=\Phi_i(\vct{\zeta}_i,\vct{q}_i(\pi_i))\) for \(\pi_i\in\Pi_i(k)\).
The unreduced full-action-space objective before safety admission and frontier pruning is
\begin{equation}
   \begin{aligned}
      \textbf{P1:}\quad
      \min_{\{\pi_i\}_{i=1}^{N}} \quad &
      \sum_{i=1}^{N}
      J_i(\pi_i,\vct{\zeta}_i)           \\
      \text{s.t.} \quad                           &
      \pi_i\in\Pi_i^{+}(k),\ \forall\, i,\quad \sum_{i=1}^{N} s_i(\pi_i)\le S.
   \end{aligned}
   \label{eq:p1-benchmark}
\end{equation}
\textbf{P1} minimizes the expected augmented Lyapunov drift over all network actions and zero-slot continuations under the shared slot budget~\(S\).
Solving \textbf{P1} online is impractical because certificate construction, routing, retransmission limits, and slot allocation remain coupled in \(\Pi_i^{+}(k)\).
\textbf{P1} therefore serves as the full action-space benchmark.
Zero-slot continuation remains available only to hold-admissible loops; loops that receive no selected network action and cannot hold enter the controlled-invariant local safe subsystem after the allocation step.

The online problem compares each admitted network action with the certified outcome when the loop receives no network allocation.
For a safe-feasible action \(\pi\), define \(J(\pi,\vct{\zeta})\triangleq \Phi(\vct{\zeta},\vct{q}(\pi))\).
The certified zero-allocation baseline is
\begin{equation}
   B_i(\vct{\zeta}_i)
   \triangleq
   \begin{cases}
      J_i^0(\vct{\zeta}_i),               & \text{if loop } i \text{ is hold-admissible}, \\
      J_i^{\mathrm{safe}}(\vct{\zeta}_i), & \text{otherwise},
   \end{cases}
   \label{eq:allocation-baseline}
\end{equation}
where \(J_i^{\mathrm{safe}}(\vct{\zeta}_i)\) is the deterministic Lyapunov drift under local safe mode.
The network-action value, measuring Lyapunov-drift reduction relative to the zero-allocation baseline, is
\begin{equation}
   U_i(\pi)\triangleq B_i(\vct{\zeta}_i)-J_i(\pi,\vct{\zeta}_i).
\end{equation}

\subsection{Slot-Constrained Dynamic Programming}

Let \(0\) denote the zero-allocation choice, with \(s_i(0)=0\) and \(U_i(0)=0\) for every loop.
If loop~\(i\) is hold-admissible, \(0\) is realized by the zero-slot continuation \(\pi_i^0\); otherwise, it is realized by local safe mode.
After safety admission, dominance pruning, and insertion of the zero-allocation choice, the reduced action set is \(\bar{\Pi}_i(k)\triangleq\Pi_i^{\mathrm{fr}}(k)\cup\{0\}\), \(i=1,\ldots,N\).
The online state-conditioned allocation problem is
\begin{equation}
   \begin{aligned}
      \textbf{P2:}\quad
      \max_{\{\pi_i\}_{i=1}^{N}} \quad &
      \sum_{i=1}^{N}
      U_i(\pi_i)                          \\
      \text{s.t.} \quad                           &
      \pi_i\in\bar{\Pi}_i(k),\ \forall\, i,\quad \sum_{i=1}^{N} s_i(\pi_i)\le S.
   \end{aligned}
   \label{eq:p2-reduced}
\end{equation}
\textbf{P2} selects frontier actions by state-conditioned harm reduction relative to the zero-allocation outcome.
Each \(\bar{\Pi}_i\) is finite, so \textbf{P2} is a multi-choice knapsack problem \cite{pisingerMinimalAlgorithm1995}.
Let \(F_i(b)\) denote the maximum total harm reduction after processing loops \(1,\ldots,i\) with at most \(b\) occupied TDMA slots.
With \(F_0(b)=0\) for \(b=0,\ldots,S\), the recurrence is
\begin{equation}
   F_i(b)
   =
   \max_{\pi\in\bar{\Pi}_i,\; s_i(\pi)\le b}
   \left\{
   F_{i-1}(b-s_i(\pi))+U_i(\pi)
   \right\}.
   \label{eq:p2-dp-recurrence}
\end{equation}
The DP records one attaining action per state \((i,b)\) and recovers the selected actions by backtracking from \(\arg\max_{0\le b\le S}F_N(b)\).
Time complexity is \(O(S\sum_{i=1}^{N}|\bar{\Pi}_i|)\); maintaining full recovery information requires \(O(NS)\) storage, while the value recursion alone uses \(O(S)\).
By Proposition~\ref{prop:frontier-sufficiency}, the reduced per-loop sets preserve the optimal frontier allocation.
Algorithm~\ref{alg:certified-frontier-knapsack} summarizes the online certified-frontier allocation rule.
Let \(K_i(k)=|\Pi_i(k)|\), with \(K_i(k)\le|\hat{\mathcal{P}}_i^{\mathrm{ul}}(k)|\,|\hat{\mathcal{P}}_i^{\mathrm{dl}}(k)|\,|\mathcal{S}_i(k)|\,|\mathcal{X}|\).
The per-cycle cost is \(O(\sum_i K_i(k)\,C_i^{\mathrm{cert}}+\sum_i K_i(k)^2+S\sum_i|\bar{\Pi}_i(k)|)\), covering certificate evaluation, dominance pruning, and DP allocation.
\begin{algorithm}[t]
   \caption{Certified-Frontier DP Allocation over TDMA Slots}
   \label{alg:certified-frontier-knapsack}
   \begin{algorithmic}[1]
      \Require Digital-twin state \(\mathcal{T}_k\), loop states \(\{\vct{\zeta}_i(k)\}_{i=1}^{N}\), slot budget \(S\)
      \Ensure Communication action or zero-allocation realization for every loop
      \For{\(i=1,\ldots,N\)}
      \State Generate candidate actions \(\Pi_i(k)\) from predicted routes, channels, blocklengths, retransmission limits, and TDMA slot allocations
      \State Admit \(\pi\in\Pi_i(k)\) into \(\Pi_i^{\mathrm{safe}}(k)\) if \(c_i<g_i(\pi)\) and \(\Phi_i(\vct{\zeta}_i(k),\vct{q}_i(\pi;k))\le0\)
      \State Compute \(B_i(\vct{\zeta}_i)\), \(J_i(\pi,\vct{\zeta}_i)\), and \(U_i(\pi)=B_i(\vct{\zeta}_i)-J_i(\pi,\vct{\zeta}_i)\)
      \State Prune dominated safe actions under \(\triangleright\); set \(\bar{\Pi}_i(k)=\Pi_i^{\mathrm{fr}}(k)\cup\{0\}\)
      \State Set choice \(0\) to nominal hold if hold-admissible, or to post-allocation local safe mode otherwise
      \EndFor
      \State Initialize \(F_0(b)=0\) for \(b=0,\ldots,S\)
      \For{\(i=1,\ldots,N\)}
      \State Compute \(F_i(b)=\max_{\pi\in\bar{\Pi}_i,\; s_i(\pi)\le b}\{F_{i-1}(b-s_i(\pi))+U_i(\pi)\}\) for all \(b=0,\ldots,S\), storing one attaining action
      \EndFor
      \State Choose \(b^\star\in\arg\max_{0\le b\le S}F_N(b)\)
      \State Backtrack stored actions from \((N,b^\star)\); realize choice \(0\) for every loop without a selected network action
   \end{algorithmic}
\end{algorithm}

\section{Closed-Loop Performance Evaluation}
\label{sec:evaluation}
\subsection{Evaluation Setup}
The evaluation uses an ns-3 closed-loop UAV swarm simulator that jointly models UAV mobility and control, multi-hop TDMA networking, scheduled slot allocation, DT-calibrated service certification, and geometry-dependent air-to-ground propagation.
Table~\ref{tab:simulation-setup} lists the common simulation parameters.
All policies execute the same urban facade-and-roof inspection mission around a building.
The UAVs fly offset facade passes and roof-corner segments, while four fixed groups rotate through cruise, hotspot, and recovery segments.
Only one group occupies the hotspot regime at a time, where waypoints use shorter legs and larger vertical excursions near facade edges and roof corners.
During hotspot visits, the simulator injects a deterministic horizontal gust realization indexed by trial and UAV, with direction resolved from the local tangent and outward building-normal basis.
This disturbance emulates short building-proximity wind bursts during facade inspection while preserving matched disturbance realizations across policies~\cite{mohamedGustsEncounteredBuildings2023}.

The radio channel uses the 3GPP TR~36.777 urban macro aerial-vehicle path-loss model with large-scale shadowing enabled and LOS/NLOS states fixed within each run.
The building is also inserted into the propagation chain as a high-loss penetration obstacle, and both data transmission and topology measurement share the same propagation and obstacle-loss model.
All policies use the same orthogonal TDMA budget with \(160\) slots of \(500~\mu\mathrm{s}\) per control cycle.
The same mission trace, channel realization model, TDMA budget, control-service horizon, and matched random seeds are used across all policies, so performance differences arise solely from service certification and slot allocation.
The service-reliability model is calibrated from deadline-level uplink, downlink, and bidirectional service outcomes under the same reference criterion, with disjoint calibration and evaluation seeds.

\begin{table}[t]
   \centering
   \caption{Key Simulation Parameters}
   \label{tab:simulation-setup}
   \setlength{\tabcolsep}{2pt}
   \begin{tabular}{@{}p{0.27\columnwidth}p{0.17\columnwidth}p{0.27\columnwidth}p{0.17\columnwidth}@{}}
      \toprule
      \textbf{Parameter} & \textbf{Value}                       & \textbf{Parameter}    & \textbf{Value}              \\
      \midrule
      UAVs               & \(50\)                               & Speed                 & \(15~\mathrm{m/s}\)         \\
      Horizon            & \(30~\mathrm{s}\)                    & Safety distance       & \(10~\mathrm{m}\)           \\
      Building footprint & \(60~\mathrm{m}\times60~\mathrm{m}\) & Control period        & \(250~\mathrm{ms}\)         \\
      Facade offset      & \(20~\mathrm{m}\)                    & Control deadline      & \(125~\mathrm{ms}\)         \\
      Altitude           & \(70\)--\(100~\mathrm{m}\)           & Certification horizon & \(250~\mathrm{ms}\)         \\
      Hotspot duration   & \(6\) cycles                         & Carrier frequency     & \(2.4~\mathrm{GHz}\)        \\
      Recovery duration  & \(4\) cycles                         & Mission type          & Facade--roof inspection     \\
      Gust duration      & \(2\)--\(6\) cycles                  & Gust acceleration     & \(2\)--\(4~\mathrm{m/s^2}\) \\
      \bottomrule
   \end{tabular}
\end{table}

\subsection{Representative-Scale Policy Comparison}

Figs.~\ref{fig:main-comparison-rmse}--\ref{fig:main-comparison-control-delay} test the central claim at \(N=50\): under the same mission, channel model, and TDMA budget, allocating certified bidirectional service to the loops where it reduces Lyapunov drift most improves closed-loop tracking.
Five policies are compared.
Fixed-Service uses a fixed regular-service selector.
Cert-Fixed applies certificate filtering with a fixed priority rule.
The dynamic-transmission-count hold-last-control baseline (DynTx-HLC)~\cite{maOptimalDynamicTransmission2022} adapts transmission count while retaining hold-last-control.
The value-of-information scheduler (VoI-Sched)~\cite{ayanAgeInformationValueInformation2019} performs two-hop scheduling by weighting information age with control-relevant state uncertainty and persistent controller-side estimate propagation.
SAFE implements the frontier allocation developed in Section~\ref{sec:frontier-allocation}.

\begin{figure}[t]
   \centering
   \includegraphics[width=\linewidth]{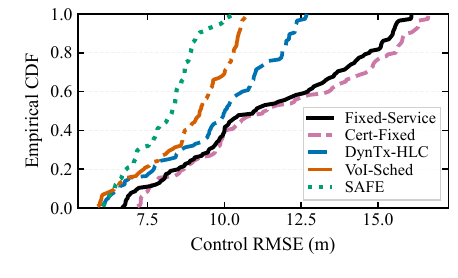}
   \caption{
      Tracking-error ECDF for the representative \(N=50\) inspection trace.
      SAFE shifts the RMSE distribution leftward relative to all four baselines.
   }
   \label{fig:main-comparison-rmse}
\end{figure}

Fig.~\ref{fig:main-comparison-rmse} provides the primary tracking-performance evidence.
The ECDF shifts leftward across the entire distribution, so the gain is not confined to easy cycles or to the median region.
SAFE reduces the mean RMSE to \(8.0~\mathrm{m}\), compared with \(11.4~\mathrm{m}\) for Fixed-Service, \(11.9~\mathrm{m}\) for Cert-Fixed, \(9.8~\mathrm{m}\) for DynTx-HLC, and \(8.9~\mathrm{m}\) for VoI-Sched.
The improvement is visible against both fixed-service rules and adaptive communication-aware baselines.
The \(95\)th-percentile RMSE is \(9.9~\mathrm{m}\), while the baselines range from \(10.6~\mathrm{m}\) to \(16.3~\mathrm{m}\).
This tail behavior matters because hotspot and gust segments create the states in which stale or poorly timed service produces the largest tracking excursions.
The RMSE evidence confirms that SAFE improves realized control performance under the same communication budget, but it does not by itself identify the underlying mechanism.

\begin{figure}[t]
   \centering
   \includegraphics[width=\linewidth]{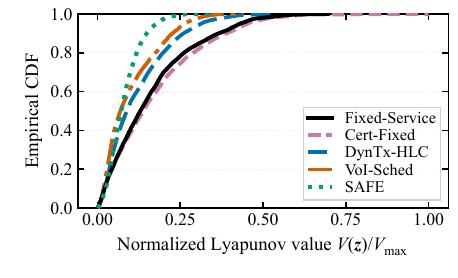}
   \caption{
      Certificate-level normalized Lyapunov-value ECDF for the same \(N=50\) trace.
      The distribution tests whether certified frontier allocation moves loop states away from the admissible-domain boundary.
   }
   \label{fig:main-comparison-lyapunov}
\end{figure}

Fig.~\ref{fig:main-comparison-lyapunov} supplies the mechanism-level evidence.
The plotted quantity is \(V(\vct{z})/V_{\max}\), where \(V_{\max}\) is the admissible-domain boundary defined in~\eqref{eq:safe-certification-domain}.
Its upper tail measures how often loop states approach the certified tracking-error boundary under the same Lyapunov matrix and safety geometry.
Relative to Fixed-Service, Cert-Fixed, DynTx-HLC, and VoI-Sched, SAFE reduces the \(95\)th-percentile normalized Lyapunov value by \(58.5\%\), \(60.4\%\), \(42.5\%\), and \(32.0\%\), respectively.
The larger reductions against Fixed-Service and Cert-Fixed show that certificate admission alone is insufficient when admitted actions are still selected by a fixed rule.
The remaining gain over DynTx-HLC and VoI-Sched shows that adapting transmission effort or scalar information value does not fully capture the bidirectional deadline, delivery, and interaction-regularity constraints embedded in the certificate.
Together with Fig.~\ref{fig:main-comparison-rmse}, this evidence links the observed tracking improvement to fewer high-risk loop states rather than only to a lower average error.

\begin{figure}[t]
   \centering
   \includegraphics[width=\linewidth]{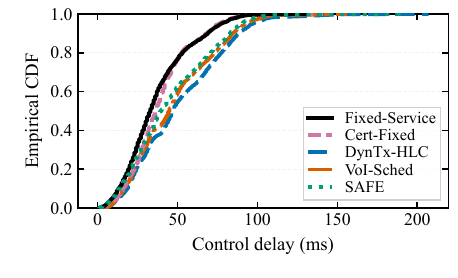}
   \caption{
      Downlink control-command delay ECDF for delivered commands.
      This diagnostic checks whether the tracking gain follows from simply reducing the delay of every delivered command.
   }
   \label{fig:main-comparison-control-delay}
\end{figure}

Fig.~\ref{fig:main-comparison-control-delay} serves as a falsification check.
If the RMSE and Lyapunov gains were solely a consequence of uniformly lower downlink delay, SAFE would also dominate the delivered-command delay distribution.
SAFE has a lower mean control-command delay than DynTx-HLC and VoI-Sched, but a higher mean delay than Fixed-Service and Cert-Fixed.
The same ordering appears at the \(95\)th percentile.
The control benefit therefore cannot be explained by uniformly reducing command delivery delay.

The observed ordering is consistent with the allocation mechanism.
Fixed-Service maintains a regular low-delay service pattern but cannot reallocate slots away from stable loops.
Cert-Fixed admits safe certificates but does not rank them by current drift reduction.
DynTx-HLC adjusts transmission effort without jointly certifying uplink estimation timeliness, downlink command delivery, and bidirectional interaction regularity.
VoI-Sched prioritizes information value through a scalar score, so it can still allocate slots to actions whose closed-loop certificate is weak for the current state.
SAFE first filters actions by envelope-level certificate safety, then ranks the surviving frontier actions by state-conditioned drift reduction \(U(\pi)\), and allocates slots to the loops with the largest hold-relative gain.

The three diagnostics form a sequential argument.
Fig.~\ref{fig:main-comparison-rmse} shows improved tracking accuracy.
Fig.~\ref{fig:main-comparison-lyapunov} shows that the improvement originates from fewer high-risk certificate states.
Fig.~\ref{fig:main-comparison-control-delay} rules out a pure communication-delay explanation, confirming that the gain arises from state-conditioned certificate-level allocation rather than from uniformly faster command delivery.

\subsection{Geometry of the Certified Supply Frontier}

Fig.~\ref{fig:certified-supply-frontier} illustrates the certificate-preserving reduction for a single loop by projecting candidate actions at a safe-envelope boundary state.
\begin{figure}[t]
   \centering
   \includegraphics[width=0.851\linewidth]{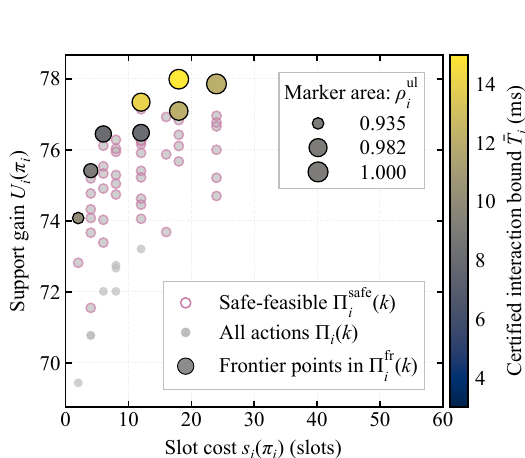}
   \caption{
      Candidate actions projected at a safe-envelope boundary state; gray markers are candidates, open overlays identify certificate-safe actions, gray-only markers are certificate-unsafe candidates, and filled frontier markers encode \(T_i^{\mathrm{cert}}\) by color and DT-calibrated \(\rho_i^{\mathrm{ul}}\) by area.
   }
   \label{fig:certified-supply-frontier}
\end{figure}
Open overlays mark the actions admitted into \(\Pi_i^{\mathrm{safe}}(k)\); gray-only candidates fail the boundary-state drift certificate; and filled markers show frontier points in \(\Pi_i^{\mathrm{fr}}(k)\).
High-gain actions already appear at low-to-moderate slot costs, while higher-cost frontier points persist because full dominance also depends on directional delays, interaction regularity, and delivery guarantees.
This geometry favors distributing slots across loops and preserving incomparable certificates before state-conditioned allocation.
Frontier pruning thus removes redundant service choices without collapsing the multi-dimensional certificate structure required by the allocator.

\begin{figure}[t]
   \centering
   \includegraphics[width=\linewidth]{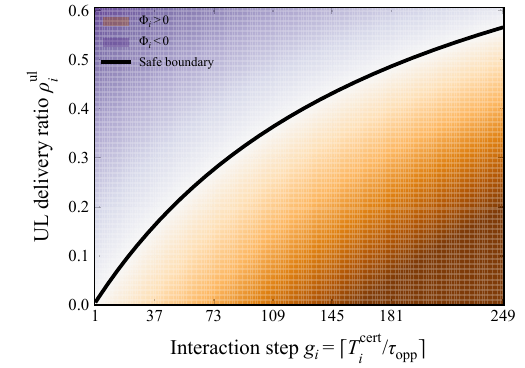}
   \caption{
      Drift-mechanism view of the shared QoS certificate.
      The plot shows the representative safe-region boundary \(\Phi_i=0\) over interaction index and uplink delivery guarantee.
   }
   \label{fig:drift-mechanism}
\end{figure}
Fig.~\ref{fig:drift-mechanism} explains why the frontier cannot be reduced to a single reliability or delay score.
The \(\Phi_i=0\) contour separates certificate-positive and certificate-negative drift for the representative state slice.
Moving to longer interaction intervals requires stronger uplink delivery before the same state remains certificate-safe, because missing telemetry over several control periods increases both state uncertainty and the hold-control drift term.
This coupling supports the design choice in Section~\ref{sec:frontier-allocation}: candidate actions must be certified jointly over directional deadlines, delivery guarantees, and interaction regularity before the allocator ranks them by \(U(\pi)\).
It also explains the gap between Cert-Fixed and SAFE in Fig.~\ref{fig:main-comparison-lyapunov}: a certificate-safe action can still be a weak allocation choice when another frontier action yields larger hold-relative drift reduction for the current state.
SAFE improves the control-facing tail precisely by selecting among safe certificates according to state-conditioned drift reduction, rather than treating all admitted certificates as equivalent.

\section{Discussion and Conclusion}
\label{sec:conclusion}
The introduction identified a missing interface between DT prediction and closed-loop scheduling: existing methods do not define when one multi-hop route--slot action safely substitutes for another from the standpoint of closed-loop stability.
This paper fills that gap by formulating DT-assisted closed-loop wireless resource allocation for low-altitude multi-hop UAV swarms as state-conditioned allocation under a shared TDMA slot budget.

The shared QoS certificate serves as the scheduler-facing interface between twin prediction and closed-loop service admissibility.
It maps predicted topology, channel conditions, route, action, and controller-side state into a five-dimensional supply specification that jointly captures bidirectional delay, DT-calibrated service reliability, interaction regularity, and slot cost, thereby making heterogeneous route--slot actions comparable through their certified control support.
The certified supply frontier then removes actions that consume more resources while providing no stronger closed-loop guarantee, and the remaining per-cycle allocation is solved exactly under the shared TDMA budget.

The theoretical contribution is the certificate partial order that makes closed-loop admissibility monotone over a finite supply space while retaining the directionality and interaction structure of the bidirectional communication service.
The practical contribution is an online allocator that converts this order into a tractable TDMA scheduling rule without collapsing heterogeneous cross-layer actions into a single scalar metric.
The evaluation confirms the central premise: DT prediction becomes useful for swarm control only after it is converted into a certificate that preserves bidirectional service semantics and state-conditioned drift value.
The control-delay falsification check further shows that the tracking gain arises from state-conditioned certificate-level allocation, not from uniformly faster command delivery.

The current formulation assumes collision-free TDMA without inter-loop spatial reuse, so slot-sum feasibility coincides with the scheduling constraint used by the online allocator.
Future work can extend this formulation along three directions: refining service-reliability calibration under real flight data, replacing scalar slot budgeting with conflict-constrained spatial-reuse placement, and constructing hierarchical certified frontiers for larger swarm topologies.

\appendices
\section{Construction of the Certified Contraction Envelope}
\label{sec:appendix-contraction-envelope}
This appendix constructs the certified one-cycle envelope \(\bar{\alpha}(\vct{q},c)\) used in the drift bound~\eqref{eq:drift-bound}.

\subsection{MJLS Parameterization and Lifted Dynamics}

For \(\chi\in\{\mathrm{ul},\mathrm{dl}\}\), let \(h^\chi=\lceil D^\chi/T_s\rceil\) denote the direction-\(\chi\) control-cycle delay index.
The certificate induced by action \(\pi\) carries the opportunity-clock interaction index \(g=g(\pi)\) defined in~\eqref{eq:tuple-interaction-cert}.
Quantizing delays on the control-cycle clock and interaction gaps on the bidirectional-opportunity clock yields a cycle-level MJLS parameterized by \((h^{\mathrm{ul}},h^{\mathrm{dl}},g)\) \cite{doCostaDiscreteTimeMarkov2005}.
The lifted state \(\vct{z}_k\in\mathbb{R}^{n(H+1)}\), with \(n=3\), stacks the current tracking error and \(H=\max(h^{\mathrm{ul}},h^{\mathrm{dl}})\) delayed copies.
Each opportunity outcome \(\xi=(\xi_{\mathrm{ul}},\xi_{\mathrm{dl}})\in\{0,1\}^2\) induces a lifted matrix \(\mtx{M}_{\xi}^{\mathrm{opp}}(h^{\mathrm{ul}},h^{\mathrm{dl}})\) with the delay-shift structure
\[
   \mtx{M}_{\xi}^{\mathrm{opp}}
   =
   \begin{bmatrix}
      \mtx{A}+\mtx{F}_0^{\xi} & \mtx{F}_1^{\xi} & \cdots    & \mtx{F}_H^{\xi} \\
      \mtx{I}_n               & \mtx{0}         & \cdots    & \mtx{0}         \\
      \mtx{0}                 & \mtx{I}_n       & \ddots    & \vdots          \\
      \vdots                  &                 & \ddots    & \mtx{0}         \\
      \mtx{0}                 & \cdots          & \mtx{I}_n & \mtx{0}
   \end{bmatrix}.
\]
The outcome-dependent top-row blocks \(\mtx{F}_h^{\xi}\) assign the LQR feedback authority \(\mtx{B}\mtx{K}\) to the available information columns.
Specifically:
\begin{itemize}
   \item \(\mtx{M}_{1,1}^{\mathrm{opp}}\) uses fresh bidirectional feedback across the current, \(h^{\mathrm{ul}}\)th, and \(h^{\mathrm{dl}}\)th columns;
   \item \(\mtx{M}_{1,0}^{\mathrm{opp}}\) removes the downlink column, shifting the unavailable command contribution to the stalest column;
   \item \(\mtx{M}_{0,1}^{\mathrm{opp}}\) removes the uplink refresh column and uses the propagated controller estimate;
   \item \(\mtx{M}_{0,0}^{\mathrm{opp}}\) removes both fresh-direction columns and concentrates the feedback at the maximum-delay column.
\end{itemize}
Each control cycle contains one bidirectional opportunity, so the cycle transition matrix equals the opportunity matrix: \(\mtx{M}_{\xi}^{\mathrm{cyc}}=\mtx{M}_{\xi}^{\mathrm{opp}}\).
The interaction index \(g\) enters through the run-length matrices \(\{\mtx{M}_{0,0}^{(j)}(h^{\mathrm{ul}},h^{\mathrm{dl}}): j=1,\ldots,g-1\}\), where \(j\) consecutive bidirectional failures place the effective feedback at the corresponding stale column.
The matrix \(\mtx{M}_{0,0}^{(j)}\) models gradual staleness growth and is generally not equal to \((\mtx{M}_{0,0}^{\mathrm{opp}})^j\).

\subsection{LMI Certificate and Contraction Factors}

Offline certification augments the opportunity outcome with the current bidirectional-failure run counter \(c\in\{0,\ldots,g-1\}\).
The finite cycle-mode set is \(\{\mtx{M}_{\xi}^{\mathrm{opp}}:\xi\in\{0,1\}^2\}\); the failure-run matrices \(\{\mtx{M}_{0,0}^{(j)}\}\) extend the admissible mode family for runs of length \(j\) with \(c<g\).

The certificate matrix \(\bmtx{X}\succ0\) is obtained from the LMI feasibility problem
\begin{equation}
   \begin{aligned}
      \min_{\bmtx{X}}\quad & \tr(\bmtx{X})                                                                                                                                            \\
      \text{s.t.}\quad     & \bmtx{X}\succeq\varepsilon\mtx{I},                                                                                                                                  \\
                           & \mtx{M}_{1,1}^{\mathrm{opp}}(h^{u},h^{d})^{\mathsf{T}}\bmtx{X}\,\mtx{M}_{1,1}^{\mathrm{opp}}(h^{u},h^{d})-\bmtx{X}\preceq-\varepsilon\mtx{I}, \\
                           & \qquad\forall\,0\le h^{u}\le\lceil D^{\mathrm{ul}}/T_s\rceil,\;0\le h^{d}\le\lceil D^{\mathrm{dl}}/T_s\rceil,
   \end{aligned}
   \label{eq:lmi-certificate}
\end{equation}
with \(\varepsilon>0\).
Only the bidirectional-success family is constrained; the remaining delivery patterns are evaluated under the same \(\bmtx{X}\), so degraded-feedback modes may be expansive.

Every raw mode factor is evaluated by the Lyapunov-weighted spectral radius
\begin{equation}
   \begin{aligned}
      \alpha_{\xi}^{\mathrm{raw}}(h^{\mathrm{ul}},h^{\mathrm{dl}})
      \triangleq
      \lambda_{\max}\Bigl(
                    &\bmtx{X}^{-1/2}\,
                    \mtx{M}_{\xi}^{\mathrm{opp}}(h^{\mathrm{ul}},h^{\mathrm{dl}})^{\mathsf{T}}\,
      \bmtx{X}
      \\
       & \quad\times
      \mtx{M}_{\xi}^{\mathrm{opp}}(h^{\mathrm{ul}},h^{\mathrm{dl}})\,
      \bmtx{X}^{-1/2}
      \Bigr)
   \end{aligned}
   \label{eq:alpha-raw}
\end{equation}
for each outcome \(\xi\in\{0,1\}^2\).
The certified contraction factor takes the pointwise maximum over the delay-index grid:
\begin{equation}
   \alpha_\xi(\vct{q})
   \triangleq
   \max_{\substack{0\le h^{\mathrm{ul}}\le\lceil D^{\mathrm{ul}}/T_s\rceil \\[2pt] 0\le h^{\mathrm{dl}}\le\lceil D^{\mathrm{dl}}/T_s\rceil}}
   \alpha_\xi^{\mathrm{raw}}(h^{\mathrm{ul}}, h^{\mathrm{dl}}).
   \label{eq:delay-supremum-envelope}
\end{equation}
The four values \(\{\alpha_{11},\alpha_{10},\alpha_{01},\alpha_{00}\}\) are the certificate-level factors from~\eqref{eq:delay-supremum-envelope}.

For a run of \(j\) consecutive bidirectional failures, define the Lyapunov-weighted growth matrix
\begin{equation}
   \begin{aligned}
      \mtx{G}_j(h^{\mathrm{ul}},h^{\mathrm{dl}})
      \triangleq {} &
      \bmtx{X}^{-1/2}\,
      (\mtx{M}_{0,0}^{(j)})^{\mathsf{T}} \\
                    & {}\times
      \bmtx{X}\,
      \mtx{M}_{0,0}^{(j)}\,
      \bmtx{X}^{-1/2},
      \qquad j=1,\ldots,g-1,
   \end{aligned}
\end{equation}
with \(\mtx{M}_{0,0}^{(j)}=\mtx{M}_{0,0}^{(j)}(h^{\mathrm{ul}},h^{\mathrm{dl}})\).
The exact-delay run-length factor is \(\alpha_{(j)}^{\mathrm{raw}}(h^{\mathrm{ul}},h^{\mathrm{dl}})\triangleq\lambda_{\max}[\mtx{G}_j(h^{\mathrm{ul}},h^{\mathrm{dl}})]\), for \(j=1,\ldots,g-1\).
The certified run-length factor is
\begin{equation}
   \alpha_{(j)}(\vct{q})
   \triangleq
   \max_{\substack{0\le h^{\mathrm{ul}}\le\lceil D^{\mathrm{ul}}/T_s\rceil \\[2pt] 0\le h^{\mathrm{dl}}\le\lceil D^{\mathrm{dl}}/T_s\rceil}}
   \alpha_{(j)}^{\mathrm{raw}}(h^{\mathrm{ul}}, h^{\mathrm{dl}}),
   \qquad j=1,\ldots,g-1.
   \label{eq:failure-run-delay-envelope}
\end{equation}
Set \(\alpha_{(0)}\triangleq1\).
After \(H\) consecutive bidirectional failures, all \(H+1\) buffer entries originate from the failure run, and each additional failure shifts the buffer without altering the transition structure.
Hence \(\mtx{M}_{0,0}^{(j)}=\mtx{M}_{0,0}^{(H)}\) for all \(j\ge H\), and the run-length factors saturate:
\begin{equation}
   \alpha_{(j)}=\alpha_{(H)},\qquad j\ge H.
   \label{eq:failure-run-saturation}
\end{equation}
When the certified failure tolerance \(g-1\) exceeds \(H\), only \(\alpha_{(1)},\ldots,\alpha_{(H)}\) require explicit computation; the remaining factors reuse \(\alpha_{(H)}\).
The admissible failure-family factor is therefore
\begin{equation}
   \max_{1\le j\le g-1}\alpha_{(j)}
   =
   \max_{1\le j\le\min(g-1,H)}\alpha_{(j)}.
   \label{eq:failure-family-contraction}
\end{equation}

\subsection{Per-Opportunity Contraction Envelope}

For one opportunity, let \(\theta_{ab}\) denote the probability of UL outcome \(a\in\{0,1\}\) and DL outcome \(b\in\{0,1\}\), where \(1\) denotes success.
The marginal constraints imply \(\hat{\rho}\le\theta_{11}\le\rho^{\wedge}\), with \(\hat{\rho}\triangleq\max\{0,\,\rho^{\mathrm{ul}}+\rho^{\mathrm{dl}}-1\}\) and \(\rho^{\wedge}\triangleq\min\{\rho^{\mathrm{ul}},\rho^{\mathrm{dl}}\}\).
Note that \(\theta_{11}\) is a same-opportunity joint-success probability, distinct from the Gilbert--Elliott persistence parameter \(p_{11}\) in~\eqref{eq:failure-transition}.
For a fixed endpoint \(b\in\{\hat{\rho},\rho^{\wedge}\}\), the induced four-outcome contraction factor is
\[
   \begin{aligned}
      E(b)
      ={} &
      b\,\alpha_{11}
      +(\rho^{\mathrm{ul}}-b)\,\alpha_{10}
      +(\rho^{\mathrm{dl}}-b)\,\alpha_{01} \\
          & {}
      +(1-\rho^{\mathrm{ul}}-\rho^{\mathrm{dl}}+b)\,\alpha_{00}.
   \end{aligned}
\]
Because \(E(b)\) is affine in \(b=\theta_{11}\), the certificate uses the larger endpoint:
\begin{equation}
   \bar{\alpha}_{\mathrm{nf}}(\vct{q})
   =
   \max\!\left\{E(\hat{\rho}),\,E(\rho^{\wedge})\right\}.
   \label{eq:per-opp-envelope}
\end{equation}
By Lemma~\ref{lem:mode-order}, \(\alpha_{11}\le\alpha_{10}\le\alpha_{00}\) and \(\alpha_{11}\le\alpha_{01}\le\alpha_{00}\).
Under this ordering, increasing \(\rho^{\mathrm{ul}}\) or \(\rho^{\mathrm{dl}}\) shifts mass from weaker toward stronger service outcomes, so \(\bar{\alpha}_{\mathrm{nf}}\) is nonincreasing in the two service-success lower bounds.

\subsection{Cycle-Level Envelope}

Each control cycle contains one bidirectional opportunity.
Conditioned on the incoming failure counter \(c<g\), the cycle decomposes into two mutually exclusive events:
\begin{itemize}
   \item \emph{Success:} the opportunity delivers bidirectional service, and the per-opportunity envelope \(\bar{\alpha}_{\mathrm{nf}}(\vct{q})\) from~\eqref{eq:per-opp-envelope} bounds the contraction factor.
   \item \emph{Failure:} the opportunity fails bidirectionally, extending the run to \(c+1\), and the run-length factor \(\alpha_{(c+1)}(\vct{q})\) from~\eqref{eq:failure-run-delay-envelope} bounds the contraction factor.
\end{itemize}
Taking the worse case gives the certified one-cycle envelope
\begin{equation}
   \bar{\alpha}(\vct{q},c)
   =
   \max\bigl\{
   \bar{\alpha}_{\mathrm{nf}}(\vct{q}),\;
   \alpha_{(c+1)}(\vct{q})
   \bigr\},
   \qquad c=0,\ldots,g-1.
   \label{eq:cycle-envelope-r1}
\end{equation}

\section{Proof of the Mode-Order Property}
\label{sec:appendix-mode-order}
\begin{proof}
   Under \(\mtx{A}=\mtx{I}_3\), the lifted matrices \(\mtx{M}_{\xi}^{\mathrm{opp}}(h^{\mathrm{ul}},h^{\mathrm{dl}})\) share an identical shift structure and differ only in the top-row blocks \(\mtx{F}_0^{\xi},\ldots,\mtx{F}_H^{\xi}\), which assign the feedback authority \(\mtx{B}\mtx{K}\) to columns indexed by the delay of the available directional samples.
   In \(\mtx{M}_{1,1}^{\mathrm{opp}}\), the uplink and downlink columns at indices \(h^{\mathrm{ul}}\) and \(h^{\mathrm{dl}}\) carry the \(\mtx{B}\mtx{K}\) contributions.
   Removing one directional success---say downlink---transfers the \(\mtx{B}\mtx{K}\) weight from column \(h^{\mathrm{dl}}\) to the maximum-delay column~\(H\), yielding \(\mtx{M}_{1,0}^{\mathrm{opp}}\).

   Write \(\mtx{M}_{1,0}^{\mathrm{opp}}=\mtx{M}_{1,1}^{\mathrm{opp}}+\bmtx{\Delta}\), where \(\bmtx{\Delta}\) is nonzero only in the top row.
   Because \(\mtx{A}=\mtx{I}_3\), the delayed tracking-error components \(\vct{e}(k-j)\) differ from \(\vct{e}(k)\) only by accumulated process noise.
   Shifting the feedback from a fresher column (\(h^{\mathrm{dl}}\le H\)) to the stalest column (\(H\)) replaces a lower-variance estimate with a higher-variance one, so the expected one-step Lyapunov growth cannot decrease.

   Formally, for any \(\bmtx{X}\succ0\) satisfying~\eqref{eq:lmi-certificate} and any delay grid point \((h^{\mathrm{ul}},h^{\mathrm{dl}})\):
   \[
      \alpha_{1,1}^{\mathrm{raw}}(h^{\mathrm{ul}},h^{\mathrm{dl}})
      \le
      \alpha_{1,0}^{\mathrm{raw}}(h^{\mathrm{ul}},h^{\mathrm{dl}})
      \le
      \alpha_{0,0}^{\mathrm{raw}}(h^{\mathrm{ul}},h^{\mathrm{dl}}),
   \]
   and the same chain holds with subscripts \((1,1)\le(0,1)\le(0,0)\).
   Taking the pointwise maximum over the delay grid via~\eqref{eq:delay-supremum-envelope} preserves the ordering at the certificate level.
   The two one-sided factors \(\alpha_{10}\) and \(\alpha_{01}\) involve different column transfers and are not ordered in general.
\end{proof}

\section{Derivation of the Augmented Lyapunov Drift Bound}
\label{sec:appendix-drift-bound}

Let \(W(\vct{z}_k,\bmtx{\Sigma}^{\mathrm{p}}_k)=V(\vct{z}_k)+\lambda_\Sigma\,\tr(\bmtx{\Sigma}^{\mathrm{p}}_k)\).
Condition on the digital-twin state \(\mathcal{T}_k\), the filtration \(\mathcal{F}_k\), and an action \(\pi\) with certificate \(\vct{q}=\vct{q}(\pi)\).
Write \(\mathbb{E}_{k,q}[\cdot]\triangleq\mathbb{E}[\cdot\mid\mathcal{F}_k,\vct{q}]\).
The control-cycle drift decomposes as
\begin{equation}
   \mathbb{E}_{k,q}[\Delta W]
   =
   \mathbb{E}_{k,q}[\Delta V(\vct{z})]
   +
   \lambda_\Sigma\,
   \mathbb{E}_{k,q}[\Delta\tr(\bmtx{\Sigma}^{\mathrm{p}})].
   \label{eq:drift-decomposition}
\end{equation}

\paragraph{Lyapunov term.}
For incoming counter \(c=c(\vct{\zeta}_k)\), certified interaction event, and single-opportunity outcome \(\xi\in\{0,1\}^2\), the cycle-level MJLS gives
\begin{equation}
   V(\vct{z}_{k+1})
   \le
   \alpha_{\xi}\,
   V(\vct{z}_k).
   \label{eq:v-step-contraction}
\end{equation}
Taking the worst admissible expectation over distributions satisfying the marginal delivery constraints and the Gilbert--Elliott persistence bounds yields
\begin{equation}
   \mathbb{E}_{k,q}[\Delta V(\vct{z})]
   \le
   \bigl(\bar{\alpha}(\vct{q},c(\vct{\zeta}_k))-1\bigr)\,V(\vct{z}_k).
   \label{eq:drift-v-bound}
\end{equation}

\paragraph{Covariance term.}
Let \(A_k\) denote the event that at least one uplink succeeds during cycle~\(k\).
On \(A_k\), which occurs with probability at least \(\rho^{\mathrm{ul}}\), the controller propagates the received telemetry sample by \(h^{\mathrm{ul}}\) steps to obtain
\begin{equation}
   \tr(\bmtx{\Sigma}^{\mathrm{p}}_{k+1})
   =
   \tr(\bmtx{\Sigma}_{\mathrm{meas}})+h^{\mathrm{ul}}\,\tr(\bmtx{\Sigma}_{\mathrm{proc}})
   =
   \bar{\Sigma}^{\mathrm{ul}}(\vct{q}).
   \label{eq:covariance-success}
\end{equation}
On \(\bar{A}_k\), which occurs with probability at most \(1-\rho^{\mathrm{ul}}\), open-loop propagation gives
\begin{equation}
   \tr(\bmtx{\Sigma}^{\mathrm{p}}_{k+1})
   =
   \tr(\bmtx{\Sigma}^{\mathrm{p}}_k)+\tr(\bmtx{\Sigma}_{\mathrm{proc}}).
   \label{eq:covariance-failure-run}
\end{equation}
Hence
\begin{equation}
   \mathbb{E}[\Delta\tr(\bmtx{\Sigma}^{\mathrm{p}})\mid \mathcal{F}_k,\bar{A}_k,\vct{q}]
   =
   \tr(\bmtx{\Sigma}_{\mathrm{proc}}).
   \label{eq:covariance-failure-bound}
\end{equation}

\paragraph{Endpoint envelope.}
Let \(p=\Pr\{A_k\mid\mathcal{F}_k,\vct{q}\}\).
Since only the lower bound \(p\ge\rho^{\mathrm{ul}}\) is certified and the expected covariance change is affine in \(p\), the maximum over \(p\in[\rho^{\mathrm{ul}},1]\) is attained at the endpoint determined by the sign of \(\bar{\Sigma}^{\mathrm{ul}}(\vct{q})-\tr(\bmtx{\Sigma}^{\mathrm{p}}_k)-\tr(\bmtx{\Sigma}_{\mathrm{proc}})\):
\begin{equation}
   \mathbb{E}_{k,q}[\Delta\tr(\bmtx{\Sigma}^{\mathrm{p}}_k)]
   \le
   \Psi_{\Sigma}(\vct{\zeta}_k,\vct{q}),
   \label{eq:covariance-drift-final}
\end{equation}
where \(\Psi_{\Sigma}\) is defined in~\eqref{eq:covariance-endpoint-envelope}.
Substituting~\eqref{eq:drift-v-bound} and~\eqref{eq:covariance-drift-final} into~\eqref{eq:drift-decomposition} yields \(\Phi(\vct{\zeta},\vct{q})\) in~\eqref{eq:drift-bound}.
The bound uses the piecewise covariance envelope and the certified contraction envelope \(\bar{\alpha}\) from~\eqref{eq:cycle-envelope-r1}; hence \(\Phi\le0\) is sufficient for the certified-event drift inequality.

\section{Proof of Drift-Certificate Monotonicity}
\label{sec:appendix-drift-monotonicity}
\begin{proof}
   Let \(\vct{q}\succeq\tilde{\vct{q}}\).
   Then \(h^\chi\le\tilde{h}^\chi\), \(g\le\tilde{g}\), and \(\rho^\chi\ge\tilde{\rho}^\chi\).
   For any admissible state with \(c<g\), shorter delay bounds \(D^\chi\le\tilde{D}^\chi\) restrict the maximization grid to a subset: \([0,\lceil D^{\mathrm{ul}}/T_s\rceil]\times[0,\lceil D^{\mathrm{dl}}/T_s\rceil]\subseteq[0,\lceil \tilde{D}^{\mathrm{ul}}/T_s\rceil]\times[0,\lceil \tilde{D}^{\mathrm{dl}}/T_s\rceil]\).
   The raw factors \(\alpha_\xi^{\mathrm{raw}}\) and \(\alpha_{(j)}^{\mathrm{raw}}\) are nondecreasing in each delay index, so the pointwise maximum over the smaller grid is no larger.
   Hence every opportunity factor and failure-run factor for \(\vct{q}\) is no larger than the corresponding factor for \(\tilde{\vct{q}}\).
   Lemma~\ref{lem:mode-order} ensures that the endpoint envelope is nonincreasing in \(\rho^{\mathrm{ul}}\) and \(\rho^{\mathrm{dl}}\); hence \(\bar{\alpha}_{\mathrm{nf}}(\vct{q})\le\bar{\alpha}_{\mathrm{nf}}(\tilde{\vct{q}})\).
   From~\eqref{eq:cycle-envelope-r1}, \(\bar{\alpha}(\vct{q},c)=\max\{\bar{\alpha}_{\mathrm{nf}}(\vct{q}),\alpha_{(c+1)}(\vct{q})\}\).
   Since both \(\bar{\alpha}_{\mathrm{nf}}(\vct{q})\le\bar{\alpha}_{\mathrm{nf}}(\tilde{\vct{q}})\) and \(\alpha_{(c+1)}(\vct{q})\le\alpha_{(c+1)}(\tilde{\vct{q}})\), the componentwise maximum gives \(\bar{\alpha}(\vct{q},c)\le\bar{\alpha}(\tilde{\vct{q}},c)\).
   Therefore the Lyapunov term of~\eqref{eq:drift-bound} is no larger under \(\vct{q}\).
   For the covariance term, \(h^{\mathrm{ul}}\le\tilde{h}^{\mathrm{ul}}\) gives \(\bar{\Sigma}^{\mathrm{ul}}(\vct{q})\le\bar{\Sigma}^{\mathrm{ul}}(\tilde{\vct{q}})\).
   In both branches of the piecewise envelope~\eqref{eq:covariance-endpoint-envelope}, a smaller \(\bar{\Sigma}^{\mathrm{ul}}\) and a larger \(\rho^{\mathrm{ul}}\) each reduce \(\Psi_{\Sigma}\), so \(\Psi_{\Sigma}(\vct{\zeta},\vct{q})\le\Psi_{\Sigma}(\vct{\zeta},\tilde{\vct{q}})\).
   Summing the Lyapunov and covariance terms proves \(\Phi(\vct{\zeta},\vct{q})\le \Phi(\vct{\zeta},\tilde{\vct{q}})\).
\end{proof}

\section{Proof of Frontier Sufficiency}
\label{sec:appendix-frontier-sufficiency}
\begin{proof}
   First, if \(\vct{q}(\pi)\succeq\vct{q}^*\) and \(\Phi(\vct{\zeta}(k),\vct{q}^*)\le0\), Proposition~\ref{prop:drift-monotonicity} gives
   \(\Phi(\vct{\zeta}(k),\vct{q}(\pi))\le\Phi(\vct{\zeta}(k),\vct{q}^*)\le0\).
   A stronger certificate therefore preserves state-conditioned admission.
   Second, let \(\pi,\pi'\in\Pi_i^{\mathrm{safe}}(k)\) with \(\pi'\triangleright\pi\).
   Then \(\vct{q}_i(\pi')\succeq\vct{q}_i(\pi)\) and \(s_i(\pi')\le s_i(\pi)\).
   By Proposition~\ref{prop:drift-monotonicity}, \(J_i(\pi',\vct{\zeta}_i)\le J_i(\pi,\vct{\zeta}_i)\), so
   \(U_i(\pi')=B_i(\vct{\zeta}_i)-J_i(\pi',\vct{\zeta}_i)\ge B_i(\vct{\zeta}_i)-J_i(\pi,\vct{\zeta}_i)=U_i(\pi)\).
   Finally, any feasible solution of \textbf{P2} that selects a dominated action \(\pi\notin\Pi_i^{\mathrm{fr}}(k)\) can replace it by its dominator \(\pi'\) without decreasing the objective or increasing slot usage.
   The replacement remains schedulable because collision-free TDMA feasibility depends only on \(\sum_i s_i(\pi_i)\le S\); absolute slot positions do not affect certificate semantics.
   Repeating this replacement over the finite action set yields an optimal solution composed entirely of the reduced action set \(\bar{\Pi}_i(k)\).
   The dynamic program in~\eqref{eq:p2-dp-recurrence} enumerates exactly the resulting per-loop multi-choice allocations under budget~\(S\), so it returns a globally optimal frontier allocation without requiring any dominated action.
\end{proof}

\section{Multi-Hop Applicability of the Failure-Chain Model}
\label{sec:appendix-failure-chain}

The two-state Gilbert--Elliott model in Section~\ref{sec:certificate-admission} is conservative for a multi-hop path when its persistence parameter upper-bounds the path-level deadline-level service-failure tail.
The digital twin supplies the calibrated path-level pair \((\mu_1,p_{11}^{\mathrm{path}})\) from bidirectional service-outcome sequences indexed by route, schedule, action, and conditioning state class.
Equation~\eqref{eq:markov-failure-tail} then upper-bounds the probability of \(L\) consecutive end-to-end bidirectional failures; heterogeneous hop coherence can be accommodated by replacing this abstraction with a multi-state service-failure model.

{\bibliographystyle{IEEEtran} \bibliography{IEEEabrv,references}}

\begin{thebibliography}{10}
\providecommand{\url}[1]{#1}
\csname url@samestyle\endcsname
\providecommand{\newblock}{\relax}
\providecommand{\bibinfo}[2]{#2}
\providecommand{\BIBentrySTDinterwordspacing}{\spaceskip=0pt\relax}
\providecommand{\BIBentryALTinterwordstretchfactor}{4}
\providecommand{\BIBentryALTinterwordspacing}{\spaceskip=\fontdimen2\font plus
\BIBentryALTinterwordstretchfactor\fontdimen3\font minus \fontdimen4\font\relax}
\providecommand{\BIBforeignlanguage}[2]{{%
\expandafter\ifx\csname l@#1\endcsname\relax
\typeout{** WARNING: IEEEtran.bst: No hyphenation pattern has been}%
\typeout{** loaded for the language `#1'. Using the pattern for}%
\typeout{** the default language instead.}%
\else
\language=\csname l@#1\endcsname
\fi
#2}}
\providecommand{\BIBdecl}{\relax}
\BIBdecl
\renewcommand{\BIBentryALTinterwordstretchfactor}{3}

\bibitem{heemelsNetworkedControlSystems2010}
W.~P. M.~H. Heemels \emph{et~al.}, ``Networked control systems with communication constraints: Tradeoffs between transmission intervals, delays and performance,'' \emph{IEEE Transactions on Automatic Control}, vol.~55, no.~8, pp. 1781--1796, 2010.

\bibitem{mengCommunicationSensingControl2026}
Z.~Meng \emph{et~al.}, ``Communication, sensing and control integrated closed-loop system: modeling, control design and resource allocation,'' \emph{Science China Information Sciences}, vol.~69, no.~3, p. 132301, 2026.

\bibitem{zhouIntegratedSensingCommunication2025}
Y.~Zhou \emph{et~al.}, ``Integrated sensing, communication, and control driven multi-agv closed-loop control,'' \emph{IEEE Transactions on Vehicular Technology}, vol.~74, no.~7, pp. 10\,853--10\,868, 2025.

\bibitem{wangTowardRealizationLowAltitude2025}
Y.~Wang \emph{et~al.}, ``Toward realization of low-altitude economy networks: Core architecture, integrated technologies, and future directions,'' \emph{IEEE Transactions on Cognitive Communications and Networking}, vol.~11, no.~5, pp. 2788--2820, 2025.

\bibitem{chaoComputingPowerSky2025}
W.~Chao \emph{et~al.}, ``Computing power in the sky: Digital twin-assisted collaborative computing with multi-uav networks,'' \emph{IEEE Transactions on Vehicular Technology}, vol.~74, no.~9, pp. 14\,466--14\,482, 2025.

\bibitem{wuDigitalTwinNetworks2021}
Y.~Wu \emph{et~al.}, ``{Digital Twin Networks}: A survey,'' \emph{IEEE Internet of Things Journal}, vol.~8, no.~18, pp. 13\,789--13\,804, 2021.

\bibitem{khanDigitalTwinWireless2022}
L.~U. Khan \emph{et~al.}, ``{Digital Twin} of wireless systems: Overview, taxonomy, challenges, and opportunities,'' \emph{IEEE Communications Surveys \& Tutorials}, vol.~24, no.~4, pp. 2230--2254, 2022.

\bibitem{zhangDigitalNetworkTwins2024}
Z.~Zhang \emph{et~al.}, ``{Digital Network Twins} for next-generation wireless: Creation, optimization, and challenges,'' \emph{IEEE Network}, pp. 1--9, 2025.

\bibitem{bellavistaApplicationDriven2021}
P.~Bellavista \emph{et~al.}, ``Application-driven network-aware {Digital Twin} management in industrial edge environments,'' \emph{IEEE Transactions on Industrial Informatics}, vol.~17, no.~11, pp. 7791--7801, 2021.

\bibitem{luAdaptiveEdgeAssociation2021}
Y.~Lu \emph{et~al.}, ``Adaptive edge association for wireless {Digital Twin} networks in {6G},'' \emph{IEEE Internet of Things Journal}, vol.~8, no.~22, pp. 16\,219--16\,230, 2021.

\bibitem{yangJointCommunicationComputation2024}
Z.~Yang \emph{et~al.}, ``A joint communication and computation framework for {Digital Twin} over wireless networks,'' \emph{IEEE Journal of Selected Topics in Signal Processing}, vol.~18, no.~1, pp. 6--17, 2024.

\bibitem{xuTaskOriented2023}
Y.~Xu \emph{et~al.}, ``Task-oriented semantics-aware communication for wireless {UAV} control and command transmission,'' \emph{IEEE Communications Letters}, vol.~27, no.~8, pp. 2232--2236, 2023.

\bibitem{leiVoIDrivenJoint2025}
L.~Lei \emph{et~al.}, ``{VoI}-driven joint optimization of control and communication in vehicular digital twin network,'' \emph{IEEE Network}, vol.~39, no.~5, pp. 155--164, 2025.

\bibitem{fangSensingCommunicationComputingControl2025}
X.~Fang \emph{et~al.}, ``Sensing-communication-computing-control closed-loop optimization for {6G} digital twin-empowered robotic systems,'' \emph{IEEE Journal on Selected Areas in Communications}, vol.~43, no.~10, pp. 3330--3346, 2025.

\bibitem{wangJointRoutingScheduling2024}
C.~Wang \emph{et~al.}, ``Joint routing and scheduling for deterministic transmission in drone swarm systems,'' in \emph{2024 10th International Conference on Computer and Communications (ICCC)}.\hskip 1em plus 0.5em minus 0.4em\relax IEEE, 2024, pp. 2510--2515.

\bibitem{liuLatencyRateReliability2021}
W.~Liu \emph{et~al.}, ``On the latency, rate, and reliability tradeoff in wireless networked control systems for iiot,'' \emph{IEEE Internet of Things Journal}, vol.~8, no.~2, pp. 723--733, 2021.

\bibitem{maOptimalDynamicTransmission2022}
Y.~Ma \emph{et~al.}, ``Optimal dynamic transmission scheduling for wireless networked control systems,'' \emph{IEEE Transactions on Control Systems Technology}, vol.~30, no.~6, pp. 2360--2376, 2022.

\bibitem{huangOptimalDownlinkUplink2020}
K.~Huang \emph{et~al.}, ``Optimal downlink--uplink scheduling of wireless networked control for industrial iot,'' \emph{IEEE Internet of Things Journal}, vol.~7, no.~3, pp. 1756--1772, 2020.

\bibitem{ayanAgeInformationValueInformation2019}
O.~Ayan \emph{et~al.}, ``Age-of-information vs. value-of-information scheduling for cellular networked control systems,'' in \emph{Proceedings of the 10th ACM/IEEE International Conference on Cyber-Physical Systems}.\hskip 1em plus 0.5em minus 0.4em\relax ACM, 2019, pp. 109--117.

\bibitem{girgisSemanticLogicalCommunicationControl2023}
A.~M. Girgis \emph{et~al.}, ``Semantic and logical communication-control codesign for correlated dynamical systems,'' \emph{IEEE Internet of Things Journal}, vol.~11, no.~7, pp. 12\,631--12\,648, 2024.

\bibitem{caoGoalOrientedCommunication2025}
J.~Cao \emph{et~al.}, ``Goal-oriented communication, estimation, and control over bidirectional wireless links,'' \emph{IEEE Transactions on Communications}, vol.~73, no.~5, pp. 3031--3045, 2025.

\bibitem{pangCommunicationControlCodesignLargeScale2025}
G.~Pang \emph{et~al.}, ``Communication-control codesign for large-scale wireless networked control systems,'' \emph{IEEE Journal on Selected Areas in Communications}, vol.~43, no.~10, pp. 3295--3312, 2025.

\bibitem{yangMultipolicyDeepReinforcement2024}
S.~Y. Yang \emph{et~al.}, ``A multipolicy deep reinforcement learning approach for multiobjective joint routing and scheduling in deterministic networks,'' \emph{IEEE Internet of Things Journal}, vol.~11, no.~10, pp. 17\,402--17\,418, 2024.

\bibitem{liMeanFieldGameTheoretic2021}
T.~Li \emph{et~al.}, ``A mean field game-theoretic cross-layer optimization for multi-hop swarm uav communications,'' \emph{Journal of Communications and Networks}, vol.~24, no.~1, pp. 68--82, 2021.

\bibitem{bekmezciFlyingAdHocNetworks2013}
I.~Bekmezci \emph{et~al.}, ``Flying ad-hoc networks (fanets): A survey,'' \emph{Ad Hoc Networks}, vol.~11, no.~3, pp. 1254--1270, 2013.

\bibitem{andersonOptimalControlLinearQuadratic1989}
B.~D.~O. Anderson \emph{et~al.}, \emph{Optimal Control: Linear Quadratic Methods}.\hskip 1em plus 0.5em minus 0.4em\relax Englewood Cliffs, NJ: Prentice-Hall, 1989.

\bibitem{shaferTutorialConformal2008}
G.~Shafer \emph{et~al.}, ``A tutorial on conformal prediction,'' \emph{Journal of Machine Learning Research}, vol.~9, pp. 371--421, 2008.

\bibitem{benTalRobustConvex1998}
A.~Ben-Tal \emph{et~al.}, ``Robust convex optimization,'' \emph{Mathematics of Operations Research}, vol.~23, no.~4, pp. 769--805, 1998.

\bibitem{dapengwuEffectiveCapacityWireless2003}
{Dapeng Wu} \emph{et~al.}, ``Effective capacity: {{A}} wireless link model for support of quality of service,'' \emph{IEEE Transactions on Wireless Communications}, vol.~2, no.~5, pp. 630--643, 2003.

\bibitem{gilbertCapacityBurstNoise1960}
E.~N. Gilbert, ``Capacity of a burst-noise channel,'' \emph{Bell System Technical Journal}, vol.~39, no.~5, pp. 1253--1265, 1960.

\bibitem{elliottEstimatesErrorRates1963}
E.~O. Elliott, ``Estimates of error rates for codes on burst-noise channels,'' \emph{Bell System Technical Journal}, vol.~42, no.~5, pp. 1977--1997, 1963.

\bibitem{doCostaDiscreteTimeMarkov2005}
O.~L.~V. do~Costa \emph{et~al.}, \emph{Discrete-Time Markov Jump Linear Systems}, ser. Probability and Its Applications.\hskip 1em plus 0.5em minus 0.4em\relax Springer London, 2005.

\bibitem{pisingerMinimalAlgorithm1995}
D.~Pisinger, ``A minimal algorithm for the multiple-choice knapsack problem,'' \emph{European Journal of Operational Research}, vol.~83, no.~2, pp. 394--410, 1995.

\bibitem{mohamedGustsEncounteredBuildings2023}
A.~Mohamed \emph{et~al.}, ``Gusts encountered by flying vehicles in proximity to buildings,'' \emph{Drones}, vol.~7, no.~1, p.~22, 2023.

\end{thebibliography}

\end{document}